\documentclass[11pt]{article}

\newif\ifblind
\blindfalse

\newif\ifdraft
\draftfalse

\usepackage{rpmacros}
\usepackage[dvipsnames,svgnames,x11names,hyperref]{xcolor}
\usepackage{stmaryrd}
\usepackage{amsthm}
\usepackage{xfrac}
\usepackage{mathpazo}
\usepackage{gitinfo}
\usepackage{enumitem}
\usepackage{fullpage}
\usepackage{thmtools}
\usepackage[T1]{fontenc} 
\usepackage[colorlinks]{hyperref}
\hypersetup{
  linkcolor=[rgb]{0,0,0.4},
  citecolor=[rgb]{0, 0.4, 0},
  urlcolor=[rgb]{0.6, 0, 0}
}
\usepackage{lipsum}
\usepackage{calrsfs,bbm}
\usepackage{multirow}
\usepackage{setspace}

\usepackage{microtype}

\usepackage[linesnumbered,vlined,ruled]{algorithm2e}
\usepackage{mathtools}

\usepackage{amsthm}
\usepackage{thmtools,thm-restate}
\usepackage[nameinlink,capitalise]{cleveref}

\numberwithin{equation}{section}
\declaretheoremstyle[bodyfont=\it,qed=\qedsymbol]{noproofstyle}

\declaretheorem[name=Observation,numbered=no]{observation*}

\declaretheorem[numberlike=equation]{theorem}
\declaretheorem[numberlike=equation,style=noproofstyle,name=Theorem]{theoremwp}
\declaretheorem[name=Theorem,numbered=no]{theorem*}

\declaretheorem[numberlike=equation]{lemma}
\declaretheorem[name=Lemma,numbered=no]{lemma*}
\declaretheorem[numberlike=equation,style=noproofstyle,name=Lemma]{lemmawp}

\declaretheorem[numberlike=equation]{corollary}
\declaretheorem[name=Corollary,numbered=no]{corollary*}
\declaretheorem[numberlike=equation,style=noproofstyle,name=Corollary]{corollarywp}

\declaretheorem[name=Proposition,numbered=no]{proposition*}

\declaretheorem[name=Claim,numbered=no]{claim*}

\declaretheorem[name=Conjecture,numbered=no]{conjecture*}

\declaretheorem[name=Question,numbered=no]{question*}

\declaretheoremstyle[bodyfont=\it,qed=$\lozenge$]{defstyle} 

\declaretheorem[numberlike=equation,style=defstyle]{definition}
\declaretheorem[unnumbered,name=Definition,style=defstyle]{definition*}

\declaretheorem[unnumbered,name=Example,style=defstyle]{example*}

\declaretheorem[unnumbered,name=Notation=defstyle]{notation*}

\declaretheorem[unnumbered,name=Construction,style=defstyle]{construction*}

\declaretheorem[numberlike=equation,style=defstyle]{remark}
\declaretheorem[unnumbered,name=Remark,style=defstyle]{remark*}

\renewcommand{\phi}{\varphi}
\renewcommand{\epsilon}{\varepsilon}

\newcommand{\size}{\operatorname{size}}
\newcommand{\depth}{\operatorname{depth}}

\newcommand{\SP}{\Sigma\Pi}
\newcommand{\SPsize}[1]{(\SP)^k\operatorname{-size}}

\newcommand{\DivTest}{\operatorname{DivTest}}
\newcommand{\hasse}[2]{\operatorname{D}^{(#1)}_{#2}}

\usepackage{nth}
\usepackage{intcalc}
\usepackage{etoolbox}
\usepackage{xstring}

\usepackage{ifpdf}
\ifpdf
\else
\usepackage[quadpoints=false]{hypdvips}
\fi

\newcommand{\ECCCurl}[2]{http://eccc.hpi-web.de/report/\ifnumcomp{#1}{>}{93}{19}{20}#1/#2/}

\newcommand{\shortECCC}[2]{\texttt{\href{http://eccc.hpi-web.de/report/\ifnumcomp{#1}{>}{93}{19}{20}#1/#2/}{eccc:TR#1-#2}}}

\newcommand{\parseECCC}[1]{
\StrSubstitute{#1}{TR}{}[\tmpstring]%
\IfSubStr{\tmpstring}{/}{ 
\StrBefore{\tmpstring}{/}[\ecccyear]%
\StrBehind{\tmpstring}{/}[\ecccreport]%
}{
\StrBefore{\tmpstring}{-}[\ecccyear]%
\StrBehind{\tmpstring}{-}[\ecccreport]%
}%
\shortECCC{\ecccyear}{\ecccreport}}

\newcommand{\homog}{\operatorname{Hom}}
\newcommand*\samethanks[1][\value{footnote}]{\footnotemark[#1]}

\newcounter{todo}

\ifdraft
\newcommand{\RPnote}[1]{\refstepcounter{todo}\textcolor{WildStrawberry}{\guillemotleft RP: #1 \guillemotright\addcontentsline{tod}{subsection}{[RP]~#1}}}
\newcommand{\MKnote}[1]{\refstepcounter{todo}\textcolor{BlueGreen}{\guillemotleft Mrinal: #1 \guillemotright\addcontentsline{tod}{subsection}{[MK]~#1}}}
\newcommand{\VRnote}[1]{\refstepcounter{todo}\textcolor{Blue}{\guillemotleft VR: #1 \guillemotright\addcontentsline{tod}{subsection}{[VR]~#1}}}
\newcommand{\SBnote}[1]{\refstepcounter{todo}\textcolor{OliveGreen}{\guillemotleft SB: #1 \guillemotright\addcontentsline{tod}{subsection}{[SB]~#1}}}
\newcommand{\SSnote}[1]{\refstepcounter{todo}\textcolor{BrickRed}{\guillemotleft SS: #1 \guillemotright\addcontentsline{tod}{subsection}{[SS]~#1}}}

\newcommand{\gitinfonotecolour}{Gray}
\newcommand{\easteregg}{}
\else
\newcommand{\RPnote}[1]{}
\newcommand{\MKnote}[1]{}
\newcommand{\VRnote}[1]{}
\newcommand{\SBnote}[1]{}
\newcommand{\SSnote}[1]{}

\newcommand{\gitinfonotecolour}{white}
\newcommand{\easteregg}{It should have been called ``The `Implicit' Function Theorem''}
\fi

\newcommand{\ignore}[1]{}
\newcommand{\gitinfonote}{git info:~\gitAbbrevHash\;,\;(\gitAuthorIsoDate)\; \;\gitVtag}

\newcommand{\Res}[3]{\ensuremath{\operatorname{Res}_{#1}(#2,#3)}} 

\newcommand{\K}{\ensuremath{\mathbb{K}}}

\newcommand{\coeff}[2]{[#1]\inbrace{#2}}
\newcommand{\coeffvec}[2]{\ensuremath{\overline{\operatorname{coeff}}_{#1}(#2)}}
\newcommand{\diag}{\mathcal{D}}

\newcommand{\esym}{\mathbf{Esym}}

\makeatletter

\newcommand\listtodoname{List of todos}
\newcommand\listoftodos{%
  \section*{\listtodoname}\@starttoc{tod}}
\makeatother

\allowdisplaybreaks
\title{Closure under factorization from a result of Furstenberg} 

\ifblind
\else
\author{
     {Somnath Bhattacharjee \thanks{University of Toronto, Canada. Email: \texttt{somnath.bhattacharjee@mail.utoronto.ca}. Research partially supported by an NSERC Discovery Grant}}
     \and
     {Mrinal Kumar \thanks{Tata Institute of Fundamental Research, Mumbai, India. Email: \texttt{\{mrinal, shanthanu.rai, varun.ramanathan, ramprasad\}@tifr.res.in}.  Research supported by the Department of Atomic Energy, Government of India, under project number RTI400112, and in part by Google and SERB Research Grants. }}
     \and
     {Shanthanu S. Rai{\samethanks[2]}}
     \and
     {Varun Ramanathan{\samethanks[2]}}
     \and
     {Ramprasad Saptharishi{\samethanks[2]}}
     \and
     {Shubhangi Saraf \thanks{University of Toronto, Canada. Email: \texttt{shubhangi.saraf@utoronto.ca}. Research partially supported by the McLean Award and an NSERC Discovery Grant.
}}
}
\fi
\date{}

\begin{document}

\maketitle

\begin{abstract}
    We show that algebraic formulas and constant-depth circuits are \emph{closed} under taking factors. In other words, we show that if a multivariate polynomial over a field of characteristic zero has a small constant-depth circuit or formula, then all its factors can be computed by small constant-depth circuits or formulas respectively. 
    
    Our result turns out to be an elementary consequence of a fundamental and surprising result of Furstenberg from the 1960s, which gives a non-iterative description of the power series roots of a bivariate polynomial. Combined with standard structural ideas in algebraic complexity, we observe that this theorem yields the desired closure results. 

   As applications, we get alternative (and perhaps simpler) proofs of various known results and strengthen the quantitative bounds in some of them. This includes a unified proof of known closure results for algebraic models (circuits, branching programs and VNP), an extension of the analysis of the Kabanets-Impagliazzo hitting set generator to formulas and constant-depth circuits, and a (significantly) simpler proof of correctness as well as stronger guarantees on the output in the subexponential time deterministic algorithm for factorization of constant-depth circuits from a recent work of Bhattacharjee, Kumar, Ramanathan, Saptharishi \& Saraf. 
\end{abstract}


\newpage 

\tableofcontents

\newpage

\section{Introduction}
This paper studies the following fundamental question --- do all factors of ``succinctly represented'' polynomial have ``succinct representations''? The answer to this question could depend on the particular model of representation. For instance, if the size of the representation is just the sum of monomials, then there are classical examples to show that $s$-sparse polynomials could have factors that are $s^{\Omega(\log s)}$ sparse. Are other ``natural models'' of computation for polynomials \emph{closed} under taking factors? 

The decade of 1980s witnessed remarkable progress for this problem for the representation of general algebraic circuits. A sequence of results \cite{Kaltofen82, K85a, GathenKaltofen85}, culminating in the celebrated results of Kaltofen \cite{K89} and Kaltofen \& Trager \cite{KT88}, showed that for any $n$-variate polynomial of degree $d$ computable by a size $s$ algebraic circuit, all its factors have algebraic circuits of size $\poly(s,d,n)$ as well. Not only that, there are randomized algorithms that take a circuit for $f$ and output circuits for all the irreducible factors of $f$ together with their multiplicities, in $\poly(s,d,n)$ time.\footnote{Strictly speaking, we need the underlying field to be rationals for this version of the result, but something very similar is true for finite fields as well.} The fact that (low degree) algebraic circuits have this highly non-trivial property is perhaps one of the strongest pieces of evidence of this model being innately natural when studying computational questions about polynomials. 

Do similar ``closure'' results hold for other natural subclasses of algebraic circuits? Indeed, such closure results for algebraic models under polynomial factorization appear to be rare. For instance, even though the last three decades or so of research in algebraic complexity has brought intense focus on the study of algebraic models, the only models where we know such closure results are for the class $\VNP$ (Chou, Kumar \& Solomon~\cite{ChouKS19}),  bounded individual degree constant-depth circuits (works of Dvir, Shpilka \& Yehudayoff~\cite{DSY09} and Oliveira~\cite{Oliveira16}), bounded individual degree sparse polynomials~\cite{BSV20} and algebraic branching programs (Sinhababu \& Thierauf~\cite{ST20}).

In addition to their considerable inherent interest, closure results for polynomial factorization for various algebraic models are also closely tied to the questions of hardness-randomness trade-offs for these algebraic models, as well as the complexity of derandomizing polynomial factorization for these models. For example, a fundamental result of Kabanets \& Impagliazzo \cite{KI04} shows that sufficiently strong lower bounds for algebraic circuits  for explicit polynomial families implies quasipolynomial time deterministic PIT algorithms for these circuits. This result crucially relies on the closure of algebraic circuits under factorization, and thus is not readily applicable for models such as formulas or constant-depth circuits. If these models were indeed closed under taking factors, then we perhaps \emph{only} need strong enough lower bounds for these models to derandomize PIT for these models. 

Similarly, a result of Kopparty, Saraf \& Shpilka \cite{KSS15} shows that given a deterministic PIT algorithm for algebraic circuits, one can derandomize Kaltofen's factorization algorithm for algebraic circuits. The main technical bottleneck for extending such a connection to other models such as formulas or constant-depth circuits is again the absence of a closure result for these models! 

In the last few years, we have had partial progress on showing such closure results for models like formulas and constant-depth algebraic circuits, and indeed even these partial results have led to extremely interesting consequences for derandomizing PIT and polynomial factorization for these models. Specifically, Chou, Kumar \& Solomon \cite{ChouKS19} showed low (but growing) degree factors of polynomials with small constant-depth circuits have non-trivially small constant-depth circuits, and used this to conclude that if we had superpolynomial lower bounds for constant-depth circuits, we would get subexponential time deterministic PIT for such circuits. Quite remarkably, such lower bounds were proved by Limaye, Srinivasan \& Tavenas \cite{LST21} a few years ago, and these lower bounds yielded non-trivial deterministic PIT for constant-depth circuits due to the results in \cite{ChouKS19}. Similarly, in recent years, we have  seen steady progress on the question of derandomizing polynomial factorization for constant-depth circuits \cite{KRS23, KRSV, DST24}, including a recent result of Bhattacharjee, Kumar, Ramanathan, Saptharishi \& Saraf \cite{BKRSS} that gives deterministic subexponential time algorithms for factorization of constant-depth circuits over fields of characteristic zero. Once again, this result crucially relies on both the results and the  techniques in the partial closure result of Chou, Kumar \& Solomon \cite{ChouKS19}. 

Thus, while the question of closure under factorization of models like formulas and constant-depth circuits has remained open, partial progress on this problem, e.g. in \cite{ChouKS19} has already had some fascinating consequences. Given this, it seems conceivable that complete closure-under-factorization results for these models would not only be of inherent interest on their own, they might also yield quantitative improvements for some of the aforementioned applications. 

\subsection{Our results}

The main result in this paper is that constant-depth algebraic circuits and algebraic formulas are closed under taking factors over fields of zero or sufficiently large characteristic. More formally, we have the following theorem. 

\begin{theorem}[Closure under factorization]\label{thm:main-intro}
Let $\F$ be a field of characteristic zero, $f$ be a polynomial on $n$ variables of degree $d$ over $\F$ and $g$ be a factor of $f$. Then, the following are true.  

If $f$ can be computed by an algebraic circuit of size $s$ and depth $\Delta$ over $\F$, then $g$ can be computed by an algebraic circuit of size $\poly(s,d,n)$ and depth $\Delta+O(1)$ over $\F$.  

If $f$ can be computed by an algebraic formula of size $s$ over $\F$, then $g$ can be computed by an algebraic formula of size $\poly(s,d,n)$ over $\F$.  
\end{theorem}

\begin{remark}
    While the statement above is stated for fields of characteristic zero, it is also true for fields of sufficiently large characteristic (depending on the degree $d$) via the same proof. 
\end{remark}

Over fields of small positive characteristic, we have the following weaker version of the above theorem. 

\begin{theorem}[Closure over finite fields of positive characteristic]\label{thm:char-p-main}
    Let $\F_q$ be a field of positive characteristic $p$, $f$ be an $n$ variate polynomial of degree $d$ over $\F_q$ and $g$ be a factor of $f$, such that the largest power of $g$ that divides $f$ is  $p^\ell\cdot e$ where $\gcd(p,e) = 1$. 

If $f$ can be computed by an algebraic circuit of size $s$ and depth $\Delta$ over $\F_q$, then $g^{p^{\ell}}$ can be computed by an algebraic circuit of size $\poly(s,d,n)$ and depth $\Delta+O(1)$ over the algebraic closure of $\F_q$.  

If $f$ can be computed by an algebraic formula of size $s$ over $\F_q$, then $g^{p^{\ell}}$ can be computed by an algebraic formula of size $\poly(s,d,n)$ over the algebraic closure of $\F_q$.  
\end{theorem}
\begin{remark}
    \autoref{thm:char-p-main} is weaker than \autoref{thm:main-intro} in two aspects - (a) we only have a circuit/formula for a power of the factor $g$ and not $g$ itself and (b) the circuit/formula for the power of  $g$ is over the algebraic closure of the base field. 

    For the rest of the paper, we just focus on the case of fields of characteristic zero. The same ideas essentially extend to the case of fields of positive characteristic with small technical changes and we discuss this case in the appendix for completeness. 
\end{remark}

To an extent, \autoref{thm:main-intro} answers some very natural open questions asked in recent years in the polynomial factorization literature (over fields of sufficiently large or zero characteristic). This includes the question of natural subclasses of algebraic circuits being closed under taking factors (Questions 1.2 and 3.1 in \cite{Forbes-Shpilka-survey} and also an open question in \cite{KSS15}) and the question of proving a non-trivial upper bound on the complexity of the factors of sparse polynomials in any algebraic model (Question 1.3 in the survey of Forbes \& Shpilka \cite{Forbes-Shpilka-survey}). As we discuss in more detail in \autoref{subsec:applications-overview},  when combined with the ideas in the recent work of Bhattacharjee et al. \cite{BKRSS}, \autoref{thm:main-intro} gives an efficient and deterministic reduction from the question of deterministic factoring and in particular deterministic irreducibility testing of polynomials computed by formulas and constant-depth circuits to the question of blackbox deterministic PIT for formulas and constant-depth circuits respectively. The general question of relationship between  derandomization of polynomial factorization for algebraic models and PIT (both in the whitebox and blackbox settings) for them was mentioned as an open problem (Question 4.1) in \cite{Forbes-Shpilka-survey}. For general algebraic circuits, such a result was shown by Kopparty, Saraf \& Shpilka \cite{KSS15}.

As alluded to in the introduction, since the closure result for general circuits has many interesting applications, perhaps one can expect some applications of the closure result in \autoref{thm:main-intro}. This indeed turns out to be the case.\footnote{As we will see in the proofs, many of these applications are in fact a consequence of the intermediate statements in the proof of \autoref{thm:main-intro}.} We now discuss these applications. 

\subsubsection{Applications}\label{subsec:applications-overview}

\paragraph*{A unified proof of closure results:} The proof of \autoref{thm:main-intro} is essentially a single unified proof of most of the closure results for factorization that we know. The proof extends over all algebraic models with some simple properties --- models should support operations such as taking products and sums, extracting homogeneous components, interpolation, etc., without significant cost. Algebraic circuits, branching programs, formulas, constant-depth circuits and exponential sums over algebraic circuits (polynomials in the class $\VNP$) are robust enough to satisfy these properties, and hence closure for them follows from the proof of \autoref{thm:main-intro}. We stress the fact that almost nothing changes in the argument as we try to infer the closure of these models under factorization.  

In addition to these closure results, we also get an alternative and perhaps slightly simpler proof that shows that factors of degree $d$ of a size $s$ circuit (of potentially exponential degree) is in the border of a circuit of size $\poly(s,d)$. This border version of Kaltofen's factor conjecture was originally proved by B\"urgisser~\cite{Burgisser04}. The proof in this paper seems to differ from the original proof conceptually, and the appearance of the notion of border complexity here happens fairly naturally. 

\paragraph*{Improved hardness-randomness trade-offs:} \autoref{thm:main-intro} immediately implies that the hardness-randomness trade-off of Kabanets and Impagliazzo \cite{KI04} also holds for models like formulas and constant-depth circuits. In particular, exponential lower bounds for constant-depth circuits for explicit polynomial families imply quasipolynomial time deterministic PIT for constant-depth circuits. 

Previously, only weaker statements of this form were known. Chou, Kumar \& Solomon \cite{ChouKS19} showed that hardness of low-degree explicit polynomial families for constant-depth circuits gives non-trivial PIT for such  circuits, and Andrews \& Forbes \cite{AF22} constructed an alternative way of using hardness of the symbolic determinant to get non-trivial PIT for these circuits. It is unclear to us if either of these routes implies a quasipolynomial time deterministic PIT, when the hardness assumption is somewhat stronger (and yet weak enough that we do not get hardness for general algebraic circuits from the depth reduction results). For instance, one concrete conclusion that can be obtained from \autoref{thm:main-intro} here is that if we have an explicit $n$-variate degree $n$ polynomial family ${P_n}$, such that any depth $\Delta$ circuit for $P_n$ has size $n^{n^{\epsilon}}$ for any $\epsilon > 0$, then we have deterministic quasipolynomial time PIT for circuits of size $\poly(n)$ and depth $\Delta + O(1)$. 

\paragraph*{Deterministic factorization of constant-depth circuits:} A recent work of Bhattacharjee et al. \cite{BKRSS} gave a deterministic subexponential time algorithm for factoring constant-depth circuits. \autoref{thm:main-intro} and the techniques therein improve the results in \cite{BKRSS} in a few aspects. 
\begin{itemize}
    \item \textit{A simpler and modular proof of correctness - }  the proof of correctness of the algorithms in \cite{BKRSS} turns out to be fairly technical. At a high level, \cite{KSS15} show that Kaltofen's result can be derandomized if we can solve PIT for certain identities that are built using the various factors of the circuit. Even though it was not known that factors of constant-depth circuits have constant-depth circuits, \cite{BKRSS} managed to show that the hitting set generator obtained by combining the lower bounds in \cite{LST21} and the hardness-randomness trade-offs in \cite{ChouKS19} (which are typically intended to be used against constant-depth circuits) do preserve nonzeroness of these identities that do not appear to be constant depth. This part of the proof relies on finer details of the structure of power series roots obtained via Newton Iteration, and the structure of the specific hitting set generator. 
    
    The techniques in the proof of \autoref{thm:main-intro} give a clean and direct proof of the correctness of the algorithms in \cite{BKRSS}, and essentially demystefy and give a more satisfying reason for why the \cite{LST21} plus \cite{ChouKS19} hitting set generator works in the algorithm in \cite{BKRSS}--- factors of constant-depth circuits \emph{are} indeed constant depth, and hence so are the relevant identities involved in the above sketch. Thus, any hitting set generator for constant-depth circuits will preserve irreducibility and factorization pattern of such circuits as well. 

    \item \textit{Constant-depth circuits for factors as output - } The algorithm in \cite{BKRSS} outputs polynomial size circuits for the irreducible factors of the input circuit. However, these output circuits need not be of constant depth (naturally, since it was not even known if there exist constant-depth circuits for them). \autoref{thm:main-intro} implies that these output polynomials have small constant-depth circuits. We show that the ideas in the proof of \autoref{thm:main-intro} can be combined with the algorithm in \cite{BKRSS} to output constant-depth circuit for all the factors. 
\end{itemize}

Additionally, stronger lower bounds for constant-depth circuit would translate to faster deterministic algorithms for polynomial identity testing, and this, in turn, would translate to faster deterministic factorization algorithms for constant-depth circuits. 
    
\paragraph*{Blackbox PIT and deterministic factorization:} A result of Kopparty, Saraf and Shpilka \cite{KSS15} showed that deterministic PIT algorithms for general circuits (in both the blackbox and the whitebox models) implies a deterministic factorization algorithm for such circuits. We can now extend this connection in the blackbox setting to models such as constant-depth circuits or formulas. However, at the moment, it is unclear to us if such a conclusion also follows from whitebox derandomization of PIT for these models.

\paragraph{Randomized algorithms for factorization:} By sampling random points instead of using a deterministic PIT algorithm in the aforementioned reduction, we also obtain an efficient randomized algorithm that takes a polynomial with a small formula / constant-depth circuit as input and outputs small formulas / constant-depth circuits for all the irreducible factors of the input, along with their multiplicities. 

We now move on to a discussion of the main techniques. 

\subsection{Proof overview}

\subsubsection*{The difficulty of proving closure for weaker models}

Before discussing the main techniques in the proof of \autoref{thm:main-intro}, we start with a brief discussion about the technical difficulty of proving close results for models like constant-depth circuits and formulas. 

One of the main ingredients of all the multivariate factorization algorithms and closure results is the notion of Newton Iteration or one of its variants (e.g. Hensel Lifting). For this discussion, we confine ourselves to Newton iteration. Let us assume that the input polynomial $P$ is of the form $P(\vecx, y) = (y - f(\vecx)) \cdot Q(\vecx, y)$ for some unknown $f$ and $Q$, with $f(\veczero) = 0$ and $\partial_y P(\veczero, 0) \neq 0$. Let $f_0(\vecx)= f(\veczero)$. We first start by observing that $f_0(\vecx) = 0$ satisfies $P(\vecx, f_0(\vecx)) = 0 \bmod{\inangle{\vecx}}$ and iteratively define
\[
f_{i+1} = f_i - \frac{P(\vecx, f_i)}{\partial_{y}P(\veczero, 0)} \bmod{\inangle{\vecx}^{i+1}},
\]
and show that $f_i$ satisfies $P(\vecx, f_i(\vecx)) = 0\bmod{\inangle{\vecx}^{i+1}}$, and (in some sense) is unique. Once $i$ exceeds $\deg(f)$, then uniqueness would guarantee that $f_i(\vecx)$ is (essentially) $f(\vecx)$. 

This is a strategy that works for general circuits, but since this process is quite sequential and involves successive composition of $P$ with itself, we were unable to show that $f_i$'s were computable by constant-depth circuits even if $P$ was. Similar issues also arise in algorithms that use Hensel lifting. Almost all closure results \cite{K89, KT88, Burgisser04, DSY09, KSS15, Oliveira16, ChouKS19, ST20,DSS22-closure} follow the above overall sketch. 

One way of getting around these issues would be to argue about the structure of these approximations of these power series roots directly without relying on the explicit iterative process used to construct them. This is essentially how the proof of \autoref{thm:main-intro} proceeds. We now discuss this in more detail. 

\subsubsection*{Main ideas}

All our results stem from the following fundamental (and surprising) result of Furstenberg from the 1960s, that essentially gives a ``closed form'' expression for the Newton iteration process, which we state now. The following statement is a special case of \cref{thm:furstenberg} for roots of multiplicity $1$. 

\begin{theorem}[\cite{Furstenberg67}]\label{thm:furstenberg-intro}
   Let $\F$ be an arbitrary field and let $P(t,y) \in \F\insquare{t,y}$ be a polynomial and $\varphi(t){\in \F\indsquare{t}}$ with $\varphi(0) = 0$ be a power series satisfying $P(t, y) = (y - \phi(t)) \cdot Q(t,y)$
    and $Q(0,0) \neq 0$. Then
    \[
        \varphi = \diag\inparen{\frac{y^2 \cdot \partial_yP(ty,y)}{P(ty,y)}}, 
    \]
    where, the \emph{diagonal} operator $\diag$ operates on a bivariate power series $F(t, y) = \sum_{i \geq 0, j \geq 0} F_{i,j}t^iy^j$ as follows:
    \[
    \diag(F)(t) := \sum_{i\geq 0} F_{i,i} \cdot t^i.
    \]
\end{theorem}

Semantically, under some mild conditions, the above theorem \emph{almost} gives a way of writing a power series root of a bivariate polynomial $P$ as a ratio of two polynomials whose complexity is close to that of the complexity of $P$. Here, the \emph{almost} part hides the complexity of computing the diagonal of a power series. 

It is worth stressing that the proof of the above result is completely elementary, and a full proof is provided in \cref{sec:explicit-formulas-for-implicit-roots} for completeness. Also, via standard transformations, the above can be simplified to the following corollary.  
\begin{corollary}
    \label{cor:furstenberg-alternative}
    Let $\F$ be an arbitrary field and let $P(t,y) \in \F\insquare{t,y}$ be a polynomial and $\varphi(t){\in \F\indsquare{t}}$ with $\varphi(0) = 0$ be a power series satisfying $P(t, y) = (y - \phi(t)) \cdot Q(t,y)$ with $Q(0,0) = 1$. Then
    \begin{align*}
    \varphi(t) &= \sum_{m\geq 1} \coeff{y^{m-1}}{(1 - \partial_y P(t,y))\cdot (y - P(t,y))^m}\\
     & = \sum_{m\geq 1} \frac{1}{m} \cdot \coeff{y^{m-1}}{(y - P(t,y))^m} \quad\text{(over char. $0$ fields)}
    \end{align*}
    where $\coeff{y^a}{G(t,y)}$ refers to the coefficient of $y^a$ in the polynomial $G(t,y)$. 
\end{corollary}

\begin{remark*}
A reader familiar with techniques in enumerative combinatorics / generating functions might notice similarities with the classical \emph{Lagrange inversion formula}, and that is indeed the case. The above can also be derived from the Lagrange inversion formula. A more elaborate discussion on this connection is provided in \cref{sec:lagrange-inversion-discussion}. 
\end{remark*}

From \autoref{cor:furstenberg-alternative}, it almost immediately follows that truncations of power series roots of formulas and constant-depth circuits have small formulas and constant-depth circuits respectively over the closure of the base field. To go from power series roots (over the field closure) to general irreducible factors (over the base field) and to do so within constant depth (or formulas) requires some new observations on combining power series roots to get general irreducible factors. In particular, we rely on the fact that the transformation between elementary symmetric and power symmetric polynomials can be done within constant depth over fields of characteristic zero (or sufficiently large characteristic), as recently shown and crucially used in a work of Andrews and Wigderson \cite{AW24}. 

Given the simplicity of ideas in the proofs in this paper, perhaps the main contribution of this work is to bring \autoref{thm:furstenberg-intro} to the attention of a theoretical computer science audience, and to notice its connections and consequences for  some very natural questions in multivariate polynomial factorization and its applications.

\subsubsection{Connections to the Lagrange Inversion Formula}
\label{sec:lagrange-inversion-discussion}

Power series that are implicitly defined via functional equations have been a subject of intense study, especially in the area of enumerative combinatorics. A classical functional equation in this context is the following --- given a power series $g(y)$, find a power series $\phi(x)$ with $\phi(0) = 0$ that satisfies the equation $x \cdot g(\phi(x)) = \phi(x)$. The Lagrange Inversion Formula (from the 18th century!), for formal power series, states that 
\[
\phi(x) = \sum_{m \geq 1} \frac{1}{m} \cdot x^m \cdot \coeff{y^{m-1}}{g(y)^m}
\]
is a solution to the above (and is also unique for nonzero $g$). 
(See \cite{SuryaWarnke} for a simple proof and its applications in enumerative combinatorics.) 

To get this closer to the setup for approximate roots, suppose $P(x,y)$ is a polynomial with $\partial_y P(0,0) = 1$. Then, if $\phi(x)$ is a power series root satisfying $\phi(0) = 0$ and $P(x, \phi(x)) = 0$, then we have that $G(x, \phi(x)) = \phi(x)$ where $G(x, y) = y - P(x, y)$, which is very similar to the functional equation $x \cdot g(\phi(x)) = \phi(x)$ above. Unsurprisingly, the Lagrange Inversion Formula can be used to derive a closed form expression for $\phi(x)$. In fact, this precise question is explicitly stated as an exercise\footnote{The book also provides solutions. For this setting, one could solve for $\phi(t, x)$ satisfying $t \cdot G(x, \phi(t,x)) = \phi(t, x)$ via the Lagrange Inversion Formula to obtain $\phi(t,x) = \sum_{m\geq 1} \frac{1}{m} \cdot t^m \cdot \coeff{y^{m-1}}{G(x, y)^m}$, and set $t = 1$.} in Stanley's book~\cite{EnumCombinatoricsVol2} on Enumerative Combinatorics! 

\begin{theorem}[Exercises 5.59 in \cite{EnumCombinatoricsVol2}]
    \label{thm:flajolet-soria-intro}
Let $\F$ be a field of characteristic $0$. Suppose $\phi(x) \in \F\indsquare{x}$ with $\phi(0) = 0$. Let $G(x,y) \in \F\indsquare{x,y}$ and $\phi$ satisfies the functional equation $G(x,\phi) = \phi$. Then, 
\[
\phi(x) = \sum_{m\geq 1} \frac{1}{m} \cdot \coeff{y^{m-1}}{G(x,y)^m}
\]
(where $\coeff{y^a}{P}$ refers to the coefficient of $y^a$ in $P$). 
\end{theorem}
The above result appears to have been discovered several times (see \cite{sokal2009,gessel16} and references within for several avatars of the above statement) in the enumerative combinatorics literature with a different set of motivations and applications in mind, and appears to have evaded the gaze of the algebraic complexity theorists. The fact that the above can also be derived from Furstenberg's identity was also observed by Hu~\cite{hu2016}. \cref{cor:furstenberg-alternative} is an immediate consequence of the above by setting $G(x, y) = y - P(x, y)$. 

\subsection*{Organization} 

The rest of this paper is organized as follows. We begin with some preliminaries in algebraic complexity in \cref{sec:preliminaries} (readers familiar with standard notions in algebraic complexity can safely skip this section). \cref{sec:explicit-formulas-for-implicit-roots} presents the theorem of Furstenberg and alternate formulations, and \cref{sec:closure} uses these to prove the structural closure results for power series roots and factors. \cref{sec:deterministic-algos-for-factorisation} gives deterministic algorithms for computing factors of constant-depth circuits (and other natural subclasses of algebraic circuits). \cref{sec:applications} presents  other applications to some known structural results in the context of factorization. We discuss the extension of the closure results to finite fields of small characteristic in \autoref{sec:finite-fields-closure}.

For readers familiar with algebraic complexity, we suggest starting directly with \cref{sec:explicit-formulas-for-implicit-roots}, and referring to the preliminaries in \cref{sec:preliminaries} as and when necessary. 

\section{Preliminaries}
\label{sec:preliminaries}
\paragraph{Notation:} 
\begin{itemize}\itemsep 0pt
    \item We use bold-face letters (such as $\mathbb{F}, \mathbb{K}$) to denote fields. 
    We use $\F[x]$ to denote the polynomial ring, $\F\indsquare{x}$ to the ring of formal power series, and $\F\indparen{x}$ refer to ring of Laurent series with respect to the variable $x$ with coefficients from the field $\F$. We use $\overline{\F}$ to refer to the algebraic closure of the field $\F$. 
    \item We use boldface letters such as $\vecx$ to refer to an order tuple of variables such as $(x_1,\ldots, x_n)$. The size of the tuple would usually be clear from context. 
    \item For a polynomial $f(x) \in \F[x]$ (or more generally in $\F\indparen{x}$) and a monomial $x^n$, we use $\coeff{x^n}{f}$ to denote the coefficient of $x^n$ in $f$. For multivariate polynomials such as $f(x,y)$, we will use $\coeff{x^n}{f}$ by interpreting $f(x) \in \F[y][x]$ and extracting the coefficient of $x^n$ as a function of $y$.
    \item The notation $\homog_{d}(F)$ refers to the degree $d$ homogeneous part of $F$, and $\homog_{\leq d}(F)$ refers to the sum of all homogeneous parts of $F$ up to degree $d$ (which is sometimes also referred to as `truncating' the polynomial at degree $d$). 
    \item The model of computation for multivariate polynomials would be the standard model of algebraic circuits (which are directed acyclic graphs with internal gates labelled by $+$ and $\times$, with leaves labelled by variables or field constants, with field constant on edges). The size of a circuit $C$, denoted by $\size(C)$, would be the number of wires in the circuit. The depth of the circuit, denoted by $\depth(C)$, would be the length of the longest path from root to a leaf node. For a polynomial $f(\vecx)$, we shall use $\size(f)$ to denote the size of the smallest circuit that computes $f$. 
    
    We also briefly use the notion of \emph{border computation} which is given by a circuit $C$ with coefficients from $\F(\epsilon)$ where $\epsilon$ is a formal variable. We shall say that the circuit $C$ is a \emph{border computation} for a polynomial $f(\vecx)$ if $C = f(\vecx) + \epsilon \cdot g(\vecx, \epsilon)$ where $g(\vecx, \epsilon) \in \F[\epsilon][\vecx]$. We use $\overline{\size}(f)$ to denote the size of the smallest circuit that border computes $f$. 

    \item A polynomial $f(\vecx)$ is said to be \emph{squarefree} if there is no non-constant polynomial $g(\vecx)$ such that $g^2$ divides $f$. Extending this, if $f(\vecx) = g_1^{e_1} \cdots g_r^{e_r}$ (with each $e_i \geq 1$) is the factorization of the polynomial into distinct irreducibles, the \emph{squarefree} part of $f$ is given by $g_1 \cdots g_r$. 
    
    \item A map $\mathcal{G}:\F[\vecx] \rightarrow \F[\vecw]$ is said to be a \emph{hitting-set generator} for a class $\mathcal{C}$ of polynomials if for every $F \in \mathcal{C}$ we have that $F \circ \mathcal{G} = 0$ implies $F = 0$. The degree of the generator is $\max_i \deg(\mathcal{G}(x_i))$. The hitting-set generator is said to be \emph{explicit} if $\mathcal{G}$ can be computed efficiently. 
\end{itemize}

\subsection{Polynomial Identity Lemma}
\begin{lemma}[Polynomial Identity Lemma {\cite[Lemma 9.2.2]{guruswami-rudra-sudan-essential}}] \label{lem:SZ-PIT-lemma}
    Let $\F$ be an arbitrary field and let $P(x) \in \F[\vecx]$ be a nonzero $n$-variate polynomial of degree $d$. Let $S$ be an arbitrary subset of $\F$. Then, \[\Pr_{\veca \in S^n}[P(\veca) = 0] \leq \frac{d}{|S|}\]
\end{lemma}

\subsection{Interpolation and consequences}

The following applications of polynomial interpolation to algebraic circuits is attributed to Michael Ben-Or. 

\begin{lemmawp}[Interpolation]
    \label{lem:interpolation}
    Let $R$ be a commutative ring that contains a field $\F$ of at least $d+1$ elements, and let $\alpha_0,\ldots, \alpha_d$ be distinct elements in $\F$. Then, for every $i \in \set{0,\ldots, d}$, there exists fields elements $\beta_{i0}, \ldots, \beta_{id}$ such that for any $f(t) = f_0 + f_1 t + \cdots + f_d t^d \in R[t]$ of degree at most $d$, we have 
    \[
    \coeff{t^i}{f} = f_i = \beta_{i0} f(\alpha_0) + \cdots + \beta_{id} f(\alpha_d)\qedhere
    \]
\end{lemmawp}

\begin{corollarywp}[Standard consequences of interpolation]
    \label{cor:interpolation-consequences}
    Let $\alpha_0,\ldots, \alpha_d$ be distinct elements in  $\F$. Then, 
    \begin{enumerate}\itemsep0pt
    \item \textbf{[Partial derivatives]} If $C(\vecx, y)$ has degree $d$ in the variable $y$, then the $i$-th order partial derivative of $C$ with respect to $y$ can be expressed as an $\F[y]$-linear combination of $\setdef{C(\vecx, \alpha_j)}{j\in \set{0,\ldots, d}}$. That is, there are polynomials $\mu_0(y), \ldots, \mu_d(y)$ (not depending on $C$) of degree at most $d$ such that 
    \[
        \partial_{y^i} C(\vecx, y) = \mu_0(y) \cdot C(\vecx, \alpha_0) + \cdots + \mu_d(y) \cdot C(\vecx, \alpha_d).
    \]
    
    \item \textbf{[Homogeneous components]} Let $C(\vecx)$ be a degree $d$ polynomial. Then, for any subset $\vecx_S \subseteq \vecx$ and any $i \in [d]$, the degree $i$ homogeneous part of $C$ with respect to $\vecx_S$, denoted by $\homog_{\vecx_S, i}(C)$, can be expressed as
    \[
        \homog_{\vecx_S, i}(C) = \sum_{j=0}^d \beta_{i,j} \cdot C(\alpha_j \cdot \vecx_S, \vecx_{\overline{S}})
    \]
    for some constants $\beta_{i,j} \in \F$ (not depending on $C$). 
    \end{enumerate}
    In particular, if $C$ is computable by a size $s$, depth $\Delta$ circuit, then all the above operations yield a circuit of size $\poly(s,d)$ and depth $\Delta + O(1)$. 
\end{corollarywp}

Observe that even if $d \ll \deg(f)$, interpolating the coefficient of $t^d$ in $f(t)$ requires $\deg(f) + 1$ many evaluations. However,  we can express the coefficient of $t^d$ in $f(t)$ as the \emph{limit (or border)} of a sum of just $d+1$ evaluations of $f$. The following lemma states this formally.

\begin{lemma}[Border Interpolation]
    \label{lem:border-interpolation}
    Let $R$ be a commutative ring that contains a field $\F$ of at least $d+1$ elements, and let $\alpha_0,\ldots, \alpha_d$ be distinct elements in $\F$. Then, there exists fields elements $\beta_{0}, \ldots, \beta_{d}$ such that for any $f(t) \in R\indsquare{t}$, we have 
    \[
    \coeff{t^d}{f} =  \frac{1}{\epsilon^d} \cdot \inparen{\beta_{0} f(\epsilon\alpha_0) + \cdots + \beta_{d} f(\epsilon\alpha_d)} + O(\epsilon).
    \]
\end{lemma}
\begin{proof}
Suppose $f = f_0 + f_1 t + f_2 t^2 + \cdots$. Define $f'(t) = f_0 + f_1t + \cdots + f_d t^d$ and $f''(t) = f - f'$. Then, applying \cref{lem:interpolation} for $g(t) = f'(\epsilon \cdot t)$ and $i = d$, we get constants $\beta_0, \ldots, \beta_d$ such that 
\[
\beta_{0} f'(\epsilon\alpha_0) + \cdots + \beta_{d} f'(\epsilon\alpha_d)  = \beta_{0} g(\alpha_0) + \cdots + \beta_{d} g(\alpha_d) = \coeff{t^d}{g} = \epsilon^d \cdot \coeff{t^d}{f'} = \epsilon^d f_d\\
\]
On the other hand, $f''(\epsilon \alpha) = O(\epsilon^{d+1})$ for any constant $\alpha$. Therefore, since $f = f' + f''$, we have
\[
\frac{1}{\epsilon^d} \cdot \inparen{\beta_{0} f(\epsilon\alpha_0) + \cdots + \beta_{d} f(\epsilon\alpha_d)} = \coeff{t^d}{f'} + O(\epsilon^{d+1-d}) = f_d + O(\epsilon).\qedhere
\]
\end{proof}

\subsection{Resultant, Discriminant and Gauss Lemma}

We now recall some definitions that are standard in the factorization literature. For more details, we encourage the readers to refer to von zur Gathen and Gerhard's book on computer algebra \cite{GG13}.

\begin{definition}[Sylvester Matrix and Resultant] \label{def:Sylvester-Resultant}
    Let $\F$ be a field. Let $P(z)$ and $Q(z)$ be polynomials of degree $a\geq 1$ and $b \geq 1$ in $\F[z]$. Define a linear map $\Gamma_{P,Q}:\F^a \times \F^b \to \F^{a+b}$ that takes polynomials $A(z)$ and $B(z)$ in $\F[z]$ of degree $a-1$ and $b-1$ respectively, and maps them to $AP + BQ$, a polynomial of degree $a+b-1$.\\ 
    The Sylvester matrix of $P$ and $Q$, denoted by $\operatorname{Syl}_z(P,Q)$, is defined to be the $(a+b)\times(p+q)$ matrix for the linear map $\Gamma_{P,Q}$. \\
    The Resultant of $P$ and $Q$, denoted by $\Res{z}{P}{Q}$, is the determinant of $\operatorname{Syl}_z(P,Q)$.
\end{definition}

\begin{definition}[Discriminant]\label{def:discriminant}
    Let $P(z)$ be a polynomial over a field $\F$. The Discriminant of $P$, denoted by $\operatorname{Disc}_z(P)$, is defined as the resultant of $P$ and $\frac{\partial P}{\partial z}$.
\end{definition}

\begin{lemma}[Resultant and GCD {\cite[Corollary 6.20]{GG13}}]\label{lem:resultant-gcd}
    Let $\mathcal{R}$ be a unique factorization domain, and let $P(z), Q(z) \in \mathcal{R}[z]$ be polynomials of degree $p\geq 1$ and $q\geq 1$ respectively. Then, $\Res{z}{P}{Q} = 0 \iff \deg_z(\gcd(P,Q)) \geq 1$. Moreover, there exist polynomials $A$ and $B$ of degree $q-1$ and $p-1$ such that $AP + BQ = \Res{z}{P}{Q}$.
\end{lemma}

\begin{lemma}[Discriminant and squarefreeness{\cite[Lemma 12]{DSS22-closure}}] \label{lem:discriminant-squarefree}
    Let $\mathcal{R}$ be a unique factorization domain, and let $P(z) \in \mathcal{R}[z]$ be a polynomial of degree at least $1$. Then, $\operatorname{Disc}_{z}(P) = 0$ if and only if $P$ is squarefree i.e. every irreducible factor of $P$ has multiplicity one.
\end{lemma}

\begin{lemma}[Gauss Lemma {\cite[Section 6.2, Corollary 6.10]{GG13}}]\label{lem:gauss}
    Let $\mathcal{R}$ be a unique factorization domain and let $\mathcal{K}$ be its field of fractions. Let $P(z) \in \mathcal{R}[z]$ be a monic polynomial. Then, $P(z)$ is irreducible in $\mathcal{R}[z]$ if and only if $P(z)$ is irreducible in $\mathcal{K}[z]$.
    In particular, the factorization of a monic polynomial $P(z)$ into its irreducible factors in $\mathcal{R}[z]$ is identical to its factorization into irreducible factors in $\mathcal{K}[z]$.
\end{lemma}

\subsection{Reducing to factorizing bivariate polynomials}

For many of the factorization applications, it would be convenient to reduce the problem to a \emph{bivariate} setting. The following definition and subsequent lemmas formalize the precise transformation and their properties. In what follows, it would be convenient to imagine the set of variables instead as $(\vecx,y)$ (that is, calling the last variable as $y$ since it plays a slightly different role). We sometimes abuse notation to refer to a set $\vecx$ of variables as $(\vecx, y)$ by reusing the same names by artificially introducing a variable. 

\begin{definition}[Valid pre-processing maps for factorization]
    \label{defn:valid-pre}

    Let $\veca, \vecb \in \F^{\abs{\vecx}}$ and let $\K = \F(\vecx)$. The homomorphism $\Psi_{\veca, \vecb}:\F[\vecx,y] \rightarrow \K[t,y]$ given by 
    \begin{align*}
    \Psi_{\veca, \vecb}: x_i & \mapsto tx_i + a_i y + b_i, \quad \text{for all $i$},\\
    \Psi_{\veca, \vecb}: y & \mapsto y
    \end{align*}
    is said to be a \emph{valid pre-processing} map for a polynomial $F(\vecx,y) \in \F[\vecx,y]$ if $G(t,y) = \Psi_{\veca, \vecb}(F)$ satisfies the following properties:
    \begin{enumerate}\itemsep 0pt
    \item \label{item:defn-valid-pre:monic} The coefficient of $y^{\deg(F)}$ in $G$ is nonzero.
    \item \label{item:defn-valid-pre:sqfree} If $\tilde{G}$ be the squarefree part of $G$, then $\tilde{G}(0,y)$ must be squarefree as well. \qedhere
    \end{enumerate}
\end{definition}

\begin{lemma}[Recovering factors from pre-processed factors]
    \label{lem:valid-pre-factorisation-pattern}
If $G(t,y) = \Psi_{\veca,\vecb}(F(\vecx,y))$ for a valid pre-processing map, then $G(0,y)$ is a univariate polynomial in $y$ over the base field $\F$. Furthermore, the factors of $G(t,y) \in \F(\vecx)[t,y]$ are in one-to-one correspondence with the factors of $F(\vecx,y)$, with the inverse map $\Psi_{\veca,\vecb}^{-1}$ given by $x_i \mapsto x_i - a_i y - b_i$ and $t \mapsto 1$. 
\end{lemma}
\begin{proof}
    For any choice of $\veca, \vecb$, the map $\psi_{\veca,\vecb}: x_i \mapsto x_i + a_i y + b_i$ is an automorphism of the polynomial ring $\F[\vecx, y]$ that keeps non-constant polynomials non-constant. Therefore, for any polynomial $F(\vecx, y)$, the factors of $F(\vecx,y)$ are in one-to-one correspondence with the factors of $\psi_{\veca, \vecb}(F(\vecx, y))$. 

    If $\Psi_{\veca, \vecb}$ is a valid pre-processing map for $F(\vecx,y)$ and $G(t,y) = \Psi_{\veca, \vecb}(F) \in \F[\vecx][t,y]$ then by \cref{defn:valid-pre}~\cref{item:defn-valid-pre:monic} we have that $G(t,y)$ is monic in $y$.\footnote{Here we just mean that the coefficient lies in $\F$.} By Gauss' lemma, the  $G(t,y)$ is reducible over $\F(\vecx)$ if and only if $G(t,y)$ is reducible over $\F[\vecx]$. If $G(t,y) = G_1(t,y) \cdot G_2(t,y)$ is a non-trivial factorization, we get a non-trivial factorization of $\psi_{\veca, \vecb}(F(\vecx,y))$ by setting $t = 1$ since $G_i(t,y)$ is monic for $i=1,2$ and will remain non-trivial under the substitution $t = 1$. 

    Thus, the factors of $G(t,y)$ are in one-to-one correspondence with the factors of $F(\vecx,y)$, and we can easily obtain one from the other. 
\end{proof}

\begin{lemma}
\label{lem:valid-pre}
Let $F(\vecx,y)$ be a nonzero polynomial of degree $d$. Suppose $\veca, \vecb \in \F^{n}$ (where $n = \abs{\vecx}$) satisfy the following properties:
\begin{enumerate}\itemsep 0pt
\item \label{item:lem-valid-pre:monic} $\homog_{d}(F)(\veca) \neq 0$,
\item \label{item:lem-valid-pre:sqfree} $\operatorname{Disc}_y(\tilde{F}(a_1y + b_1, \ldots, a_ny + b_n, y)) \neq 0$ where $\tilde{F}$ is the squarefree part of $F$ and $\operatorname{Disc}_y(F)$ denotes the discriminant of $F$ with respect to $y$.
\end{enumerate}
Then, $\Psi_{\veca, \vecb}$ is a valid pre-processing map for $F(\vecx,y)$. 
\end{lemma}
\begin{proof}
It is easy to observe that the coefficient of $y^d$ in $G(t,y) = \Psi_{\veca,\vecb}(F)$ is precisely $\homog_d(F)(\veca)$. Thus, this implies \cref{defn:valid-pre} \cref{item:defn-valid-pre:monic}. 

By \cref{lem:valid-pre-factorisation-pattern}, if $\tilde{G}$ is the squarefree part of $G$, then $\tilde{G} = \Psi_{\veca, \vecb}(\tilde{F})$. Therefore, $\tilde{G}(0,y)$ will be squarefree if and only if $\tilde{F}(a_1y + b_1, \ldots, a_ny + b_n, y)$ is square free. Hence, $\operatorname{Disc}_y(\tilde{F}(a_1y + b_1, \ldots, a_ny + b_n, y)) \neq 0$ implies \cref{defn:valid-pre} \cref{item:defn-valid-pre:sqfree}.
\end{proof}

\begin{remark}\label{rem:pre-processing-maps-exist}
A consequence of the above lemma, along with the Polynomial Identity Lemma (\autoref{lem:SZ-PIT-lemma}) is that, for any polynomial $F(\vecx, y)$, the map $\Psi_{\veca,\vecb}$ when $\veca, \vecb$ is picked at random is a valid pre-processing map. In particular, valid pre-processing maps always exist. 
\end{remark}

\subsection{Factorization of monic polynomials into power series}

The following are some standard facts about power series roots of polynomials under modest conditions. 

\begin{lemma}[Factorization into power series]
    \label{lem:factorisation-into-power-series}
    Let $f(t, y) \in \K[t,y]$ be a polynomial that is monic in $y$ such that $f(0, y)$ is squarefree. For each $\alpha \in \overline{\K}$ (the algebraic closure) such that $f(0, \alpha) = 0$, there is a unique power series $\phi_\alpha(t) \in \overline{\K}\indsquare{t}$ satisfying $\phi_\alpha(0) = \alpha$ such that $f(t, \phi_\alpha(t)) = 0$. 

    \noindent
    In fact, the polynomial $f(t, y)$ factorizes in $\overline{\K}\indsquare{t}[y]$ as 
    \[
    f(t, y) = \prod_{\alpha \in Z}(y - \phi_\alpha(t))
    \]
    where $Z$ is the set of roots of $f(0, y)$ in $\overline{\K}$. 
\end{lemma}

\noindent
The above lemma is essentially folklore and \cite[Section 3]{DSS22-closure} gives a formal proof of the above.

\subsection{Some known results on algebraic circuits}

\begin{theoremwp}[Computing squarefree decomposition (Theorem I.9 in \cite{AW24})] 
    \label{thm:AW-squarefree-decomposition-computation}

    Let $\F$ be a field of characteristic zero. 
    Given an algebraic circuit of size $s$ and depth $\Delta$ computing a polynomial degree $d$ polynomial $f(\vecx)$, consider the (unique) sequence of polynomials $f_1, \ldots, f_d \in \F[\vecx]$ such that $\gcd(f_i, f_j) = 1$ and $f = \prod_{i=1}^n f_i^i$. Then, each $f_i$ can be computed using size $\poly(s,d)$ and depth $\Delta + O(1)$ circuits. 
    In particular, the squarefree part of $f$ (which is equal to $f_1 \cdots f_n$) is computable by size $\poly(s,d)$ and depth $\Delta + O(1)$ circuits. 

    Furthermore, the circuits can be computed in polynomial time given access to an oracle solving PIT for polynomial size constant-depth circuits (assuming $
    \Delta$ is constant). 
\end{theoremwp}

\begin{remark}
    The model of constant depth algebraic circuits as defined in \cite{AW24}  are allowed to have division gates. However, an inspection of their proof of \autoref{thm:AW-squarefree-decomposition-computation} reveals that an oracle access to a PIT algorithm for division free constant depth algebraic circuits is sufficient to obtain an efficient algorithm that takes as input a constant depth circuit (without divisions) and outputs constant depth algebraic circuits (again, without division gates) for its squarefree decomposition. 

    Moreover, the same argument holds for algebraic formulas, provided we have a PIT oracle for such formulas. 
\end{remark}

We also make use of known hitting set generators for polynomial size constant-depth circuits. Although \cite[Corollary 5]{LST21} is stated as a whitebox PIT, this also provides an explicit hitting set generator. The statement below is an alternate generator construction from a result of Andrews and Forbes~\cite{AF22}. 

\begin{theoremwp}[Explicit hitting sets for constant-depth circuits~(Theorem 6.8 in the full version of \cite{AF22})]
    \label{thm:hitting-sets-for-constant-depth-circuits}
    Let $\F$ be a field of characteristic zero. For every $k \in \N$, there is a hitting set generator $\mathcal{G}_k: \F[\vecx] \rightarrow \F[\vecw]$ with $\abs{\vecw} = n^{\sfrac{1}{2^k} + o(1)}$ and for the class of $\poly(n)$-size depth $\Delta \leq o(\log\log\log n)$ circuits. 
\end{theoremwp}

\section{Explicit formulas for implicitly defined power series}
\label{sec:explicit-formulas-for-implicit-roots}

For a bivariate power series $F(x,y) = \sum_{i,j} F_{i,j} x^i y^j\in \F\indsquare{x,y}$, define the \emph{diagonal} operator $\diag(F)(t)$ as
\[
\diag(F)(t) = \sum_{i\geq 0} F_{i,i} \cdot t^i.
\]
The following result of Furstenberg~\cite{Furstenberg67} allows us to express power series roots as a diagonal of a rational function. We present Furstenberg's proof to keep the exposition self-contained. 

\begin{theorem}[{\cite[Proposition 2]{Furstenberg67}}] \label{thm:furstenberg}
   Let $\F$ be an arbitrary field. Let $P(t,y) \in \F\indsquare{t,y}$ be a power series and $\varphi(t){\in \F\indsquare{t}}$ be a power series satisfying 
    \[
    P(t, y) = (y - \phi(t))^e \cdot Q(t,y)
    \]
    for some $e \geq 1$ that is invertible in $\F$. If $\phi(0) = 0$ and $Q(0,0) \neq 0$, then
    \begin{equation}
        \label{eqn: furstenberg-expression}
        \varphi = \diag\inparen{\frac{y^2 \cdot \partial_yP(ty,y)}{e \cdot P(ty,y)}}
    \end{equation}
\end{theorem}
\begin{remark*}
Although the \cite[Proposition 2]{Furstenberg67} originally stated it for $P(t,y)$ being a polynomial and for $e = 1$, it can be seen to readily extend to $P(t,y)$ being a power series as well (cf. \cite{hu2016}) and for any $e$ that is invertible in $\F$. 
\end{remark*}

\begin{proof}
    By scaling $P$ if required, we may assume without loss of generality that $Q(0,0) = 1$. 
    \begin{align*}
        P(t,y) &= (y-\varphi(t))^e Q(t,y) \\
        \implies \frac{\partial_yP(t,y)}{P(t,y)} & = \frac{e}{y - \phi(t)} + \frac{\partial_yQ(t,y)}{Q(t,y)}\\
        \implies \frac{y^2 \cdot \partial_yP(ty,y)}{e \cdot P(ty,y)} & = \frac{y^2}{y - \phi(ty)} + \frac{y^2 \cdot \partial_yQ(ty,y)}{e \cdot Q(ty,y)}\\
        \implies \diag\inparen{\frac{y^2 \cdot \partial_yP(ty,y)}{e \cdot P(ty,y)}} & = \diag\inparen{\frac{y^2}{y - \phi(ty)}} + \diag\inparen{\frac{y^2 \cdot \partial_yQ(ty,y)}{e \cdot Q(ty,y)}}. 
    \end{align*}
    We first deal with the second summand. Since $Q(0,0) = 1$, we can write $Q(ty,y)$ as $(1-\tilde{Q}(ty,y))$ for with $\tilde{Q}$ satisfying $\tilde{Q}(0,0) = 0$. Thus, the second summand, as a power series expression, becomes 
    \begin{align*}
        \diag\inparen{\frac{y^2 \partial_y Q(ty,y)}{e \cdot Q(ty,y)} } &= \diag\inparen{y^2 \cdot e^{-1} \cdot \partial_y Q(ty,y) \inparen{\sum_{i=0}^{\infty}{\tilde{Q}(ty,y)^i}}}
    \end{align*}
    In $\partial_y Q(ty,y) \cdot (\sum_{i=0}^{\infty}{\tilde{Q}(ty,y)^i})$, every monomial has $y$-degree at least as large as the $t$-degree (since we replaced $t$ by $ty$). Multiplying by $y^2$ ensures that there are no monomials with $t$-degree equal to $y$-degree. Hence, $\diag\inparen{\frac{y^2 \partial_y Q(ty,y)}{e Q(ty,y)} } = 0$.

    For the first summand, since $\phi(0) = 0$, we have $\phi(ty)$ is divisible by $y$ and hence
    \begin{align*}
        \diag \inparen{ \frac{y}{1-(\varphi(ty)/y)} } &= \diag \inparen{y \sum_{i\geq 0} \inparen{\frac{\varphi(ty)}{y}}^i} = \diag \inparen{ \sum_{i=0}^\infty \inparen{\frac{(\varphi(ty))^i}{y^{i-1}}}}.
    \end{align*}
    Observe that for every $i\geq 0$, $m_i:= (\varphi(ty))^i/y^{i-1}$ satisfies $\deg_t(m_i) = \deg_y(m_i) + i-1$, which means that $\deg_t(m_i) = \deg_y(m_i)$ if and only if $i=1$. Thus, $\diag \inparen{ \frac{y}{1-(\varphi(ty)/y)} } = \diag\inparen{\varphi(ty)} = \varphi$. 
\end{proof}

The diagonal expression above can be simplified to a slightly more convenient expression for implicitly defined power series roots (which we state below for the case of roots of multiplicity one). 

\begin{corollary}
    \label{cor:flajolet-soria-formula-for-roots}
    Let $\F$ be an arbitrary field. Let $P(t,y) \in \F\indsquare{t,y}$ and $\phi(t) \in \F\indsquare{t}$ such that $\phi(0) = 0$, $P(t, \phi(t)) = 0$ and $\partial_y P(0,0) = \alpha \neq 0$. Then, 
    \[
    \phi(t) = \sum_{m \geq 1} \frac{1}{\alpha^{m+1}} \cdot \coeff{y^{m-1}}{\partial_y P(t,y) \cdot (\alpha y - P(t,y))^m}.
    \]
    For characteristic zero fields, the following is an alternate expression
    \[
    \phi(t) = \sum_{m \geq 1} \frac{1}{m \cdot \alpha^{m}} \cdot \coeff{y^{m-1}}{(\alpha y - P(t,y))^m}.
    \]
\end{corollary}
\begin{proof}
    By scaling $P$ if required, assume $\partial_yP(0,0) = 1$. Since $P(0,0) = 0$, this implies that $\frac{P(ty,y)}{y}$ is a polynomial with constant term equal to $1$, and is thus invertible as a power series. By \cref{thm:furstenberg}, we have
    \begin{align*}
    \phi(t) & = \diag\inparen{\frac{y \cdot \partial_y P(ty,y)}{P(ty,y)/y}} = \diag\inparen{\frac{y \cdot \partial_y P(ty,y)}{1 - (1 - \frac{P(ty,y)}{y})}}\\
            & = \diag\inparen{\sum_{m\geq 0} y \cdot \partial_y P(ty,y) \cdot \inparen{1 - \frac{P(ty,y)}{y}}^m}
    \end{align*}
    The term corresponding to $m = 0$ is just $y \cdot \partial_y P(ty, y)$ and consists only of monomials where the $y$-degree is greater than the $t$-degree and hence does not contribute any diagonal terms. Hence, 
    \begin{align*}
    \phi(t) & = \diag\inparen{\sum_{m \geq 1} y \cdot \partial_y P(ty,y) \cdot \inparen{1 - \frac{P(ty,y)}{y}}^m}\\
    \implies [t^n](\phi) & = \coeff{t^ny^n}{\sum_{m \geq 1} y \cdot \partial_y P(ty,y) \cdot \inparen{1 - \frac{P(ty,y)}{y}}^m}\\
    & = \coeff{t^ny^n}{\sum_{m \geq 1} y^{1-m} \cdot \partial_y P(ty,y) \cdot \inparen{y - P(ty,y)}^m}\\
    & = \sum_{m \geq 1} \coeff{t^n y^{n + m -1}}{\partial_y P(ty,y) \cdot \inparen{y - P(ty,y)}^m}
    \end{align*}
    For any Laurent series $R(t,y)$, we have $\coeff{t^i y^j}{R(t,y)} = \coeff{t^i y^{i+j}}{R(ty, y)}$. Hence, we have
    \begin{align*}
    [t^n](\phi) & = \sum_{m \geq 1} \coeff{t^n y^{m -1}}{\partial_y P(t,y) \cdot \inparen{y - P(t,y)}^m}\\
    & = \coeff{t^n}{\sum_{m \geq 1} \coeff{y^{m -1}}{\partial_y P(t,y) \cdot \inparen{y - P(t,y)}^m}}
    \end{align*}
    which completes the proof of the first expression (the proof for $\partial_y P(0,0) = \alpha \neq 1$ follows similarly). \\

    \noindent
    For the expression over characteristic zero fields, let $G(t,y) = y - P(t,y)$. Then
    \begin{align*}
        \coeff{t^n}{\phi(t)} & = \sum_{m \geq 1} \coeff{y^{m -1}}{(1 - \partial_y G(t,y)) \cdot G(t,y)^m}\\
        & = \sum_{m \geq 1} \inparen{\coeff{y^{m - 1}}{G(t,y)^m}  - \coeff{y^{m-1}}{\partial_y G(t,y) \cdot G(t,y)^m}}\\
        & = \sum_{m \geq 1} \inparen{\coeff{y^{m - 1}}{G(t,y)^m}  - \inparen{\frac{1}{m+1}}\coeff{y^{m-1}}{\partial_y G(t,y)^{m+1}}}
    \end{align*}
    For any power series $R(x,y)$ and any $k > 0$, note that $\coeff{y^{k-1}}{\partial_y R(x,y)} = k \cdot \coeff{y^k}{R(x,y)}$. Hence, 
    \begin{align*}
        \coeff{t^n}{\phi(t)} & = \sum_{m \geq 1} \inparen{\coeff{y^{m - 1}}{G(t,y)^m}  - \inparen{\frac{m}{m+1}}\coeff{y^{m}}{G(t,y)^{m+1}}} \\
        & = \sum_{m \geq 1} \frac{1}{m} \cdot \coeff{y^{m-1}}{G(t,y)^{m}}.\qedhere
    \end{align*}
\end{proof}

\section{Closure under taking factors}
\label{sec:closure}
In this section, we use the techniques discussed in \autoref{sec:explicit-formulas-for-implicit-roots} to prove closure results (\autoref{thm:main-intro}). 

We start by proving upper bounds on the complexity of power series roots, followed by a proof of upper bounds on general factors. 

\subsection{Complexity of power series roots}
\begin{theorem}[Power series roots without multiplicity]
    \label{thm:closure-powerseries-roots}
    Let $P(\vecx,y) \in \F[\vecx,y]$ be a polynomial computed by a circuit $C$, and let $\varphi(\vecx) {\in \F\indsquare{\vecx}}$ be a power series satisfying $\varphi(\veczero) = 0$, $P(\vecx,\varphi(\vecx)) = 0$ and $\partial_y P(\veczero,0) \neq 0$. Then, for any $d \in \N$, there is a circuit $C'$ computing $\homog_{\leq d}\insquare{\varphi}$ such that
    \[\size(C') \leq \poly(d,\size(C))\]
    \[\depth(C') \leq \depth(C) + O(1)\]
\end{theorem}
\begin{proof}
    This theorem follows almost immediately from \cref{cor:flajolet-soria-formula-for-roots} along with some standard ideas in factorization literature. 
    
    For simplicity, by scaling the polynomial $P$ if required, we may assume without loss of generality that $\partial_y P(\veczero, 0) = 1$. 
    We first perform the standard transformation of replacing each $x_i$ by $t\cdot x_i$ and work over the field $\K := \F(\vecx)$. We define $\hat{P}(t,y):= P(t\cdot \vecx,y) \in \K[t,y]$ and $\hat{\varphi}(t):= \varphi(t\cdot \vecx) {\in \K\indsquare{t}}$. Thus, we maintain the conditions $\hat{\varphi}(0) = 0$, $\hat{P}(t,\hat{\varphi}(t))=0$ and $\partial_y \hat{P}(0,0) = 1$. We can now apply \cref{cor:flajolet-soria-formula-for-roots} to get
    \[
    \hat{\phi}(t) = \sum_{m\geq 1} \coeff{y^{m-1}}{\partial_y \hat{P}(t,y) \cdot (y - \hat{P}(t,y))^m}.
    \]
    
    Note that, since $\hat{P}(0,0) = 0$ and $\partial_y \hat{P}(0,0) = 1$, we have that every monomial of $y - \hat{P}(t,y)$ is either divisible by $t$ or by $y^2$. Therefore, 
    \begin{align*}
    \homog_{\leq d}\insquare{\hat{\varphi}(t)} & = \homog_{\leq d}\insquare{\sum_{m\geq 1} \coeff{y^{m-1}}{\partial_y \hat{P}(t,y) \cdot (y - \hat{P}(t,y))^m}}\\
        & = \homog_{\leq d}\insquare{\sum_{m = 1}^{2d} \coeff{y^{m-1}}{\partial_y \hat{P}(t,y) \cdot (y - \hat{P}(t,y))^m}}
    \end{align*}
    since every monomial of $(y - \hat{P}(t,y))^{\ell} = (t\cdot A + y^2 \cdot B)^{\ell}$ either has $t$-degree at least $\ell/2$, or $y$-degree at least $\ell$ and thus terms with $m > 2d$ have no contribution to the LHS.

    The expression is clearly in $\F[\vecx,t]$ and not just $\K[t]$. Setting $t=1$ helps us retrieve $\homog_{\leq d}\insquare{\varphi(\vecx)}$ since the transformation $x_i \mapsto t\cdot x_i$ ensures that each monomial has the same $\vecx$-degree and $t$-degree in $\varphi(t\cdot \vecx)$. The depth and size bounds follow via interpolation and homogenization (\cref{lem:interpolation} and \cref{cor:interpolation-consequences}). 
\end{proof}

\subsection{Complexity of general factors}

We now proceed to prove our closure result for general factors. If $f(t,y)$ is a polynomial that is monic in $y$ and is divisible by $g(t,y)$, then over the algebraic closure $\overline{\F}$ we can express $g(0,y) = \prod_{i=1}^\ell (y - \alpha_i)$. Thus, a simple proof to obtain a constant-depth circuit for $g(t,y)$ \emph{over the algebraic closure $\overline{\F}$} would be to just consider
\[
g(t,y) = \homog_{\leq d} \inparen{\prod_{i=1}^\ell (y - C_{\alpha_i}(t))}
\]
where $C_{\alpha_i}(t)$ is the root of the power series lifted from $\alpha_i$ obtained via \cref{thm:closure-powerseries-roots}. (We need a few additional ingredients, such as  reducing to the square-free case and taking a suitable shift )

To obtain a circuit for the factors over the base field, we require a few more modifications that we now describe (and obtaining a circuit over the base field would also be essential for computing the factors algorithmically). We would require the following statement from the recent work of Andrews and Wigderson~\cite{AW24} to prove our closure result for general factors. 

\begin{theorem}[Theorem I.8 of \cite{AW24}]
    \label{thm:AW-esym-of-g-of-roots}
    Let $\F$ be a field of zero or large characteristic. Suppose $f,g,h \in \F[z]$ with $\deg(f), \deg(g), \deg(h) \leq d$. Suppose $\alpha_1, \ldots, \alpha_d \in \overline{\F}$ be the roots of $f(z)$ with multiplicity, with $h(\alpha_i) \neq 0$ for all $i$. Then, for any $r \in [d]$, $\esym_r(\frac{g(\alpha_1)}{h(\alpha_1)}, \ldots, \frac{g(\alpha_d)}{h(\alpha_d)})$ can be computed by a circuit of size $\poly(d)$ and depth $O(1)$ over the coefficients of $f$, $g$ and $h$. 
\end{theorem}

It is worth stressing that the above circuit is over the base field $\F$, and \emph{does not} take the $\alpha_i$'s as input; the inputs are just the coefficients of $f$, $g$ and $h$ which come from the base field. 

\begin{theorem}[General factors of algebraic circuits]
    \label{thm:closure-general-factors}
    Let $\F$ be a field of zero or large enough characteristic, and let $P(\vecx) \in \F[\vecx]$ be a polynomial on $n$ variables of degree $d$ computed by a circuit $C$ of size $s$ and depth $\Delta$. Then, any factor $g(\vecx) \in \F[\vecx]$ of $P$ is computable by a circuit of size $\poly(s,d,n)$ and depth $\Delta + O(1)$ over $\F$. 
\end{theorem}

\begin{proof}
    We may assume that we work with the squarefree part of $P(\vecx)$, which is also computable by a circuit of size $\poly(s,d,n)$ and depth $\Delta + O(1)$ by appealing to \cref{thm:AW-squarefree-decomposition-computation}. By reusing symbols, let us assume that $P(\vecx)$ is squarefree. 

    \medskip

    By interpreting $P(\vecx)$ as an element of $\F[\vecx, y]$, let $\tilde{P}(t,y) = \Psi(P(\vecx,y))$ for some valid pre-processing map $\Psi$ (recall \cref{defn:valid-pre}, and that they always exist (\cref{rem:pre-processing-maps-exist})). 
    Note that $\tilde{P}(0,y)$ is squarefree, and we have that $\tilde{P}(t,y) \in \F[\vecx][t,y]$ is computable by a size $\poly(s,d,n)$, depth $\Delta + O(1)$ circuit over $\F$. By \cref{lem:valid-pre-factorisation-pattern}, it suffices to show that an arbitrary factor $g(t,y) \in \F[\vecx][t,y]$ of $\tilde{P}(t,y)$ is computable by a $\poly(s,d,n)$ size, depth $\Delta + O(1)$ circuit over the field $\F$. 

    Define the Laurent series $\hat{R}(z) \in \F[\vecx,t]\indparen{z}$, and its truncated rational function $R(z) \in \F[\vecx,t](z)$ as follows:

    \begin{align*}
    \hat{R}(z) & = z + \sum_{m \geq 1} \inparen{\frac{1}{\partial_y \tilde{P}(0, z)}}^{m+1} \cdot \coeff{y^{m-1}}{\partial_y \tilde{P}(t,y+z) \cdot (y \cdot \partial_y \tilde{P}(0,z) - \tilde{P}(t,y+z))^m}\\
    R(z) & = z + \sum_{m = 1}^{2d+2} \inparen{\frac{1}{\partial_y \tilde{P}(0, z)}}^{m+1} \cdot \coeff{y^{m-1}}{\partial_y \tilde{P}(t,y+z) \cdot (y \cdot \partial_y \tilde{P}(0,z) - \tilde{P}(t,y+z))^m}
    \end{align*}
    Note that we have that the rational function $R(z)$ can be easily expressed as $\frac{R_{\text{num}}(z)}{R_{\text{denom}}(z)}$ where each $R_{\text{num}}(z)$ and $R_{\text{denom}}(z)$ are both computable by a $\poly(s,d,n)$ sized depth $\Delta + O(1)$ circuits since $\tilde{P}(t,y)$ is given by a $\poly(s,d,n)$ sized depth $\Delta + O(1)$ circuit (using \cref{lem:interpolation} and \cref{cor:interpolation-consequences}). Furthermore, $R_{\text{denom}}(z) = (\partial_y \tilde{P}(0,z))^{2d+3}$ 
    and since $\tilde{P}(0,y)$ is square free we have $R_{\text{denom}}(\alpha)$ is nonzero for every root $\alpha$ of $\tilde{P}(0,y)$. \\

    For any $\alpha \in \overline{\F}$ such that $\tilde{P}(0, \alpha) = 0$, note that $\hat{R}(\alpha)$ is in fact an element of $\F[\vecx]\indsquare{t}$ and $R(\alpha)$ is an element of $\F[\vecx][t]$. Also, if $\lambda(t,y) = (y \cdot \partial_y \tilde{P}(0,\alpha) - \tilde{P}(t,y+\alpha))$, then $\lambda(0,0) = \partial_y \lambda(0,0) = 0$ and hence every monomial of $\lambda(t,y)$ is divisible by either $t$ or $y^2$. 
    Therefore, we have
    \[
    R(\alpha) = \hat{R}(\alpha) \bmod{t^{d+1}}. 
    \]
    For any such $\alpha \in \overline{\F}$, by \cref{cor:flajolet-soria-formula-for-roots} (applied to $\phi_\alpha(t) - \alpha$ being a root of $\tilde{P}(t, y + \alpha)$), we have $\phi_\alpha(t) := \hat{R}(\alpha) \in \F[\vecx]\indsquare{t}$ as the unique power series such that $\phi_\alpha(0) = \alpha$ and $\tilde{P}(t, \phi_\alpha(t)) = 0$. 

    Thus, if $\tilde{P}(0,y) = \prod_{i=1}^r (y - \alpha_i)$ for $\alpha_i \in \overline{\F}$, and $g(0,y) = \prod_{\alpha\in S} (y - \alpha)$ for some subset $S \subset \set{\alpha_1,\ldots, \alpha_r}$, then the power series roots of $g(t,y)$ are precisely $\phi_\alpha(t)$ for $\alpha \in S$ and hence
    \begin{align*}
    g(t, y) & = \prod_{\alpha \in S} (y - \phi_\alpha(t)) = \prod_{\alpha \in S} (y - \hat{R}(\alpha)). \\
           & = \prod_{\alpha \in S} (y - R(\alpha)) \bmod{t^{d+1}}\\
    \implies g(t,y) & = \homog_{\leq d}\inparen{g'(t,y)}\\
    \text{where}\quad g'(t,y) & = \prod_{\alpha \in S} (y - R(\alpha)). 
    \end{align*}
    Each coefficient of any $y^i$ in $g'(t,y)$ is an appropriate elementary symmetric polynomial of the set $\setdef{ \frac{R_{\text{num}}(\alpha)}{R_{\text{denom}}(\alpha)} }{\alpha \in S}$. Since $g(0,y) = \prod_{\alpha \in S} (y - \alpha)$, the elementary symmetric polynomials of the set $\setdef{\alpha}{\alpha \in S}$ are just the coefficients of $g(0,y)$, which are just elements of the field $\F$. By \cref{thm:AW-esym-of-g-of-roots}, the elementary symmetric polynomials of the set $\setdef{\frac{R_{\text{num}}(\alpha)}{R_{\text{denom}}(\alpha)}}{\alpha \in S}$ can computed as $\poly(s,d,n)$ sized depth $\Delta + O(1)$ circuits. Therefore, we have a similar circuit for $g'(t,y)$ and hence also for $g(t,y)$ (by \cref{cor:interpolation-consequences}). \qedhere
\end{proof}

Therefore, over any characteristic zero (or large enough characteristic) field, classes of algebraic circuits such as $\VP$, $\VNP$, algebraic branching programs, algebraic formulas, constant-depth circuits are all closed under taking factors. (See \cref{cor:KSS-blackbox-generalisation} for a slightly more detailed statement.) 

\section{Deterministic algorithms for factorization}
\label{sec:deterministic-algos-for-factorisation}

As discussed in \cref{subsec:applications-overview}, our closure results lead to a clean proof of correctness for the results of \cite{BKRSS}, which gave deterministic subexponential time algorithms to output efficient circuits (of potentially unbounded depth) for each of the factors of constant-depth circuit. Moreover, we can output constant-depth circuits for each of the factors.

The core of the algorithm is the following lemma whose proof we will defer to the end of the section. A version of this lemma also appeared in~\cite{BKRSS} but only for a specifically chosen hitting set generator. The statement below is more general, and the proof is much cleaner. 

\begin{lemma}[Irreducibility preservation] 
    \label{lem:irreducibility-preservation}
    Let $\F$ be a field of zero (or large enough) characteristic. 
    Let $F(\vecx, t, y) \in \F[\vecx, t, y]$ be a nonzero degree $d$ polynomial that is computable by a circuit of size $s$ and depth $\Delta$. Suppose $F$ is monic in $y$, with the property that $F(\vecx, 0, y) = F(\veczero, 0, y) \in \F[y]$ (i.e., every monomial divisible by an $x_i$ is also divisible by $t$), and $F(\veczero, 0, y)$ is squarefree. 

    Let $\mathcal{G}:\F[\vecx]\rightarrow \F[\vecw]$ be a hitting set generator for the class of size $\poly(s,d)$, depth $\Delta + O(1)$ circuits. 

    Then, for every irreducible factor $G(\vecx, t, y)$ of $F(\vecx,t,y)$ we have that $G \circ \mathcal{G} \in \F[\vecw, t, y]$ is also irreducible. 
\end{lemma}

We now proceed with the main theorem of this section. 

\begin{theorem}\label{thm:deterministic-factorization-constant-depth-circuits}
    Let $\F$ be the field of rational numbers. 
    Fix any constant $\Delta \in \N$ and $\varepsilon>0$. Then, there is a deterministic algorithm $\mathcal{A}_{\Delta,\varepsilon}$ that takes as input a size $s$ depth-$\Delta$ circuit for a degree $d$ polynomial $P(\vecx) \in \F[x_1, \dots, x_n]$ and outputs circuits of size $\poly(s,d)$ and depth $\Delta + O(1)$ for each irreducible factor $g(\vecx)$ of $P(\vecx)$, along with their multiplicities. Moreover, $\mathcal{A}_{\Delta,\varepsilon}$ runs in time $\poly(s,d)^{O(n^{\varepsilon})}$.
\end{theorem}
\begin{remark}
    The theorem continues to be true over any field of zero or sufficiently large characteristic assuming that we have an efficient deterministic algorithm to factor univariates over this field. 
\end{remark}
For any bivariate degree $d$  polynomial $F(t, y) \in \F[\vecx][t,y]$ that is monic in $y$, we define the polynomial $R_F(z) \in \F[\vecx, t][z]$ (as in the proof of \cref{thm:closure-general-factors}), as follows

\[
R_F(z) = z + \sum_{m = 1}^{2d+2} \inparen{\frac{1}{\partial_y F(0, z)}}^{m+1} \cdot \coeff{y^{m-1}}{\partial_y F(t,y+z) \cdot (y \cdot \partial_y F(0,z) - F(t,y+z))^m}.
\]

\medskip

\begin{proof}[Proof of \cref{thm:deterministic-factorization-constant-depth-circuits}]
Let $P(\vecx)$ be an $n$-variate degree $d$ polynomial given by a circuit of size $s$ and depth $\Delta$. We may assume without loss of generality that $P(\vecx)$ is squarefree (as the squarefree component\footnote{In fact, the algorithm from \cref{thm:AW-squarefree-decomposition-computation} outputs the \emph{squarefree decomposition} of the polynomial. The squarefree decomposition of a polynomial $P(\vecx)$ is a sequence of polynomials $(P_1, \dots, P_r)$ such that each $P_i$ is a product of exactly those irreducible factors of $P$ that have multiplicity $i$ in its factorization. In particular, \cref{thm:AW-squarefree-decomposition-computation} immediately gives us the multiplicity of each factor that we obtain from the rest of the algorithm. Since our candidate factors from the algorithm have constant-depth circuits, we can also run a divisibility test on powers of each candidate factor to compute their multiplicities.} can be extracted using \cref{thm:AW-squarefree-decomposition-computation}). 
We outline the rough steps of the algorithm $\mathcal{A}_{\Delta,\varepsilon}$ below and elaborate on the correctness. 

\begin{enumerate}
\item \textbf{(Pre-processing)} Build a circuit $C$ of size $\poly(s)$ and depth $\Delta + O(1)$ for $F(\vecx, t, y)  = \Psi_{\veca, \vecb}(P)\in \F[\vecx][t, y]$ where $\Psi_{\veca, \vecb}$ is a valid pre-processing map for $P(\vecx)$. 
\item \textbf{(Variable reduction)} For a generator $\mathcal{G}:\F[\vecx] \rightarrow \F[\vecw]$ for size $\poly(s)$, depth $\Delta + O(1)$ circuits, define the polynomial $\tilde{F}(\vecw, t, y) := F(\vecx, t, y) \circ \mathcal{G}$. (Instantiating with the generator in \cref{thm:hitting-sets-for-constant-depth-circuits}, we may assume $\abs{\vecw} = n^\epsilon$)
\item \textbf{(Factorizing variable-reduced polynomial)} Factorize $\tilde{F}(\vecw, t, y)$ into irreducibles as $\tilde{F}(\vecw, t, y) = \tilde{G_1}(\vecw, t, y) \cdots \tilde{G_k}(\vecw, t, y)$. Use interpolation to compute the coefficients of $g_i(y) = \tilde{G_i}(\vecw, 0, y) = \tilde{G_i}(\veczero, 0, y)$ for all $i \in [k]$. 
\item \textbf{(Building the factors)} For each $j \in [r]$, if $S_j \subset \overline{\F}$ are the roots of $g_j(y)$ in the algebraic closure, define the polynomial
\[
G_j(\vecx, t, y) \coloneq \homog_{\leq d}\inparen{\prod_{i \in S_j} (y - R_F(\alpha_i))}
\]
From the coefficients of $g_j(y)$, use \cref{thm:AW-esym-of-g-of-roots} to compute the coefficients (as elements of $\F[\vecx, t]$) of $G_j(\vecx, t, y)$ via $\poly(s)$ size depth $\Delta + O(1)$ circuits. 
\item \textbf{(Undo pre-processing and return)} Return $\setdef{\Psi_{\veca, \vecb}^{-1}(G_j)}{j \in [k]}$.
\end{enumerate}

\noindent
We will justify correctness for each of the above steps. 

\paragraph{Pre-processing:} By \cref{cor:interpolation-consequences}, the highest degree homogeneous part of $P(\vecx)$ is also computable by size $\poly(s)$, depth $\Delta + O(1)$ sized circuits, and by \cite{AW24} we have that $\operatorname{Disc}_y(P)$ is computable by a size $\poly(s)$, depth $\Delta + O(1)$ circuit. Thus, by \cref{lem:valid-pre}, any hitting set for size $\poly(s)$, depth $\Delta + O(1)$ circuits may be used to compute a valid pre-procesing map $\Psi_{\veca, \vecb}$. \\

\noindent
For what follows, let $F(\vecx, t, y) = \Psi_{\veca, \vecb}(P)$.

\paragraph{Variable reduction:} Note that the polynomial $F(\vecx, t, y)$ satisfy the requirements of \cref{lem:irreducibility-preservation}. Thus, by \cref{lem:irreducibility-preservation}, we have that $\mathcal{G}$ preserves the irreducibility of the irreducible factors of $F(\vecx, t, y)$. 

\paragraph{Factorizing variable-reduced polynomial:} Once we have a variable reduced polynomial, any off-the-shelf factorization algorithm (such as \cite{lecerf2007}) may be employed to factorize $F \circ \mathcal{G}$ in time $\poly((sd)^{\abs{\vecw}})$. Computing the coefficients of $g_k(y)$ can be done via interpolation (\cref{lem:interpolation}).

\paragraph{Building the factors:} Let $\tilde{F}(\vecw, t, y) = F \circ \mathcal{G} = \tilde{G_1} \cdots \tilde{G_k}$. By \cref{lem:irreducibility-preservation}, we have that $F = G_1 \cdots G_k$ is the decomposition of $F$ into irreducibles with $\tilde{G_j}(\vecw, t, y) = G_j \circ \mathcal{G}$. Consider an arbitrary irreducible factor $G_j(\vecx, t,y)$ of $F$ with coefficients over the field $\F$. By \cref{lem:BKRSS-factor-characterisation}, we have that $G_j(\vecx, t, y) = Q_U(\vecx, t, y)$ for an appropriate set $U \subseteq [r]$. Since the generator is only applied to the $\vecx$ variables, we have that $\tilde{G_j}(\vecw, 0, y) = G_j(\vecx \circ \mathcal{G}, 0, y) = G_j(\veczero, 0, y) = \prod_{i \in U}(y - \alpha_i)$. Since we have already computed $\tilde{G_j}$, we have the coefficients of $G_j(\vecx, 0, y) = G_j(\veczero, 0, y)$ which are the elementary symmetric polynomials of $\setdef{\alpha_i}{i\in U}$. As in the proof of \cref{thm:closure-general-factors}, we can use \cref{thm:AW-esym-of-g-of-roots} to compute a $\poly(s)$ size, depth $\Delta + O(1)$ circuit for $G_j$. 

\paragraph{Undo pre-processing:} Now that we have obtained the irreducible factors of $F(\vecx, t, y) = \Psi_{\veca, \vecb}(P(\vecx))$, \cref{lem:valid-pre-factorisation-pattern} provides the inverse transformation to obtain the corresponding factors of $P(\vecx)$. \\

\paragraph{Running time:} Finding the right pre-processing map $\Psi_{\veca,\vecb}$ using $\mathcal{G}$ takes time $\poly(s,d)^{O(n^{\varepsilon})}$. The factorization of the variable-reduced polynomial also runs in time $\poly(s,d)^{O(n^{\varepsilon})}$.  Rest of the steps take time $\poly(s,d)$. Thus, the total running time is $\poly(s,d)^{O(n^{\varepsilon})}$. \\

This completes the proof correctness of \cref{thm:deterministic-factorization-constant-depth-circuits} modulo the proof of \cref{lem:irreducibility-preservation}. 
\end{proof}

\subsection{Proof of {\autoref{lem:irreducibility-preservation}}}

At the core of the algorithm of Bhattacharjee et al \cite{BKRSS} was a method to characterize variable reductions that preserve the factorization structure of the polynomial $F(\vecx, t, y)$. Recall that $F(\vecx, t, y)$ is monic in $y$ and $F(\veczero, 0, y)$ is a squarefree. 

\begin{lemmawp}[Lemma 8.3 in \cite{BKRSS}] 
    \label{lem:BKRSS-factor-characterisation}
    Let $\set{\alpha_1,\ldots, \alpha_r}$ be the roots of $F(\veczero, 0, y)$ in $\overline{\F}$. For a subset $S \subset [r]$, define 
    \[
    Q_S(\vecx, t, y) = \homog_{\leq d}\inparen{\prod_{i \in S}(y - R_F(\alpha_i))}
    \]
    Then the factors of $F(\vecx, t, y)$ over the field $\F$ is exactly the same as 
    \[
        \mathcal{F} = \setdef{Q_U(\vecx, t, y)}{\begin{array}{c}U \subseteq [r] \text{ where } Q_U(\vecx, t, y) \in \overline{\F} [\vecx, t, y] \text{ divides }F \text{ and }\\Q_U(\vecx, 0, y) = Q_U(\veczero, 0, y) = \prod_{i \in U}(y - \alpha_i) \in \F[y] \end{array}}. \qedhere
    \]
\end{lemmawp}

Let $G(\vecx, t, y) \in \F[\vecx, t, y]$ be an arbitrary irreducible factor of $F(\vecx, t,y)$. By the above lemma, there exists some $U \subseteq [d]$ such that $G(\vecx, t, y) = Q_{U}(\vecx, t, y)$ with $Q_U(\veczero, 0, y)$ having coefficients in $\F$. Note that, since $F(\veczero, 0, y)$ is squarefree,  distinct elements of $\mathcal{F}$ have distinct set of roots when $\vecx, t$ are set to zero. Since $G(\vecx, t, y)$ is irreducible, the set $U$ is minimal in the sense that for every $\emptyset \neq U' \subsetneq U$, we have that $Q_{U'}(\vecx, t, y)$ has coefficients outside $\F$ or does not divide $F(\vecx, t, y)$. 

For the sake of contradiction, assume that $\tilde{G}(\vecw, t, y) := G \circ \mathcal{G}$ is reducible and $h(\vecw, t, y) \in \F[\vecw, t, y]$ is a non-trivial factor of $\tilde{G}$. Then, $h(\veczero, 0, y)$ divides $\tilde{G}(\veczero, 0, y) = G(\veczero, 0, y)$ and the set of roots of $h(\veczero, 0, y)$ in $\overline{\F}$ is $\setdef{\alpha_i}{i\in U'}$ for some $\emptyset \neq U' \subsetneq U$. 

Consider the polynomial $Q_{U'}(\vecx, t, y)$. As in the proof of \cref{thm:closure-general-factors}, since $h(\veczero, 0, y)$ has coefficients in $\F$ and the $Q_{U'}$ is symmetric with respect to the set $\setdef{\alpha_i}{i\in U'}$, we have that $Q_{U'}(\vecx, t, y)$ has all coefficients in $\F$ as well. Hence, as argued above, $Q_{U'}(\vecx, t, y)$ does not divide $F(\vecx, t, y)$. 

However, applying \cref{lem:BKRSS-factor-characterisation} for $\tilde{G}(\vecw, t, y)$, we have 
\[
h(\vecw, t, y) = \homog_{\leq d}\inparen{\prod_{i\in U'}(y - R_{\tilde{G}}(\alpha_i))}
\]
As $F(\veczero, 0, y) = \tilde{F}(\veczero, 0, y)$ is squarefree and $\tilde{F}(\veczero, 0, \alpha_i) = 0$ for each $\alpha_i \in U'$, by \cref{lem:factorisation-into-power-series} there is a unique power series root $\tilde{\phi_i}(\vecw, t)$ for $\tilde{F} \bmod{t^{d+1}}$ that satisfies $\tilde{F}(\vecw, t,\tilde{\phi_i}) = 0$ and $\tilde{\phi_i}(0) = \alpha_i$. Note that both $R_{\tilde{G}}(\alpha_i)$ and $R_F(\alpha_i) \circ \mathcal{G}$ satisfy these properties. Hence, by the uniqueness of the power series modulo $t^{d+1}$, we have

\begin{align*}
R_F(\alpha_i) \circ \mathcal{G} & = R_{\tilde{G}}(\alpha_i) \bmod{t^{d+1}}\\
\implies h(\vecw, t, y) & = \homog_{\leq d}\inparen{\prod_{i\in U'}(y - R_{\tilde{G}}(\alpha_i))}\\
& = \homog_{\leq d}\inparen{\prod_{i\in U'}(y - R_F(\alpha_i) \circ \mathcal{G})}\\
& = \homog_{\leq d}\inparen{\prod_{i\in U'}(y - R_F(\alpha_i))}  \circ \mathcal{G} = Q_{U'} \circ \mathcal{G}. 
\end{align*}

Therefore, we have that $Q_{U'}(\vecx, t, y) \in \F[\vecx, t, y]$ does not divide $F(\vecx, t, y)$ but $h(\vecw, t, y) = {Q_{U'} \circ \mathcal{G}} \in \F[\vecw, t, y]$ does divide $\tilde{F}(\vecw, t, y) = F(\vecw, t, y) \circ \mathcal{G}$. It turns out that divisiblity testing of a pair of polynomials can be reduced to an appropriate polynomial identity test. This reduction was first observed by Forbes~\cite{Forbes15} and then crucially used in deterministic factorization algorithms \cite{KRS23,KRSV, DST24, BKRSS}. We give below a lemma from \cite{BKRSS} that implements the reduction in \cite{Forbes15} via the results of \cite{AW24}. 

\begin{lemmawp}(Lemma 8.9 in \cite{BKRSS})
    \label{lem:divisibility-test-to-PIT}
    Let $D \geq t \geq 0$ be integer parameters. Let $\F$ be any field of characteristic zero or large enough. Then, there is a constant-depth $\poly(D,t)$-sized circuit $\DivTest_{D,t}$ on $D+t+1$ variables, that takes $(D+t)$ inputs labelled $f_0, \ldots, f_{D-1} \in \F$ and $g_0,\ldots, g_{t-1} \in \F$ respectively, such that 
    \[
        \DivTest_{D,t}(y,f_0,\ldots, f_{D-1}, g_0, \ldots, g_{t-1}) = 0
    \]
    if and only if the polynomial $f(y) = f_0 + f_1 y + \cdots + f_{D-1} y^{D-1} + y^D$ divides the polynomial $g(y) = g_0 + g_1 y + \cdots + g_{t-1} y^{t-1} + y^t$. 
\end{lemmawp}

\medskip

Define $C(\vecx, t, y) := \DivTest(y, \coeffvec{y}{F(\vecx,t,y)},\coeffvec{y}{Q_{U'}(\vecx,t,y)})$ where $\coeffvec{y}{F}$ refers to the vector of coefficients when $F$ is interpreted as a univariate in $y$ (with coefficients involving the other variables). By \cref{lem:divisibility-test-to-PIT}, we have that $C(\vecx, t, y)$ is a nonzero polynomial since $Q_{U'}$ does not divide $F$ but $C \circ \mathcal{G}$ is zero since $Q_{U'} \circ \mathcal{G}$ divides $F \circ \mathcal{G}$. But $C$ is a circuit of size $\poly(s)$ and depth $\Delta + O(1)$ and hence this violates the assumption that $\mathcal{G}$ is a hitting set generator for this class. Hence, we must have that $G \circ \mathcal{G}$ continues to be irreducible for every irreducible factor $G(\vecx, t, y)$ of $F(\vecx, t, y)$. \hfill \qed {\scriptsize (\cref{lem:irreducibility-preservation})}

\medskip

\begin{remark}
In \cite{BKRSS}, the polynomials $R_F(\alpha_i)$ were instead replaced by truncated power series obtained via Newton Iteration, and therefore it was not known if the polynomials $Q \in \mathcal{F}$ are computable by constant-depth circuits. As a consequence, \cite{BKRSS} could not provide a polynomial size constant depth upper bound for the above circuit $C$. Thus, \cite{BKRSS} involved a fairly delicate argument to show that the \cite{LST21}+\cite{KI04}+\cite{ChouKS19} generator maintains the nonzeroness of these non-divisibility identity tests that arises from approximate power series roots obtained via Newton Iteration. 

With \cref{thm:closure-powerseries-roots}, we now can argue that the circuit $C$ above is indeed a polynomial size constant-depth circuit and hence \emph{any} generator for this class of circuits would preserve the factorization of $F(\vecx, t, y)$.  
\end{remark}

\subsection{Deterministic factorization from hitting-set generators} 

The algorithm in \cref{thm:deterministic-factorization-constant-depth-circuits} and its analysis via \cref{lem:irreducibility-preservation} proves a more general statement. 

\begin{corollary}(Informal)
    \label{cor:KSS-blackbox-generalisation}
    Let $\F$ be a field of characteristic zero or large enough, and let $\mathcal{C}$ be a \emph{robust enough} class of circuits that is $\mathcal{C}$ is closed under small sums and products, substitution by sparse polynomials (thereby admitting interpolation). Consider the larger class $\mathcal{C}'$ computing polynomials of the form $F(g_1,\ldots, g_m)$ where $F$ is computable by a $\poly(m)$-sized constant-depth circuit, and each $g_i \in \mathcal{C}$. 

    If we have a blackbox PIT for the class $\mathcal{C}'$ running in time $T(n)$, then we have a deterministic $T(\poly(n))$-time algorithm to factorize polynomials from the class $\mathcal{C}$. 
\end{corollary}

This, for natural subclasses of algebraic circuits --- such as algebraic formulas, algebraic branching programs, algebraic circuits, etc. --- we have a reduction from factorization to polynomial identity testing. This generalizes the result of Kopparty, Saraf and Shpilka~\cite{KSS15} who established this connection for the class of general algebraic circuits. Further, by solving each PIT instance by random sampling (and using \autoref{lem:SZ-PIT-lemma}), we get an efficient randomized algorithm that takes a polynomial from $\mathcal{C}$ as input and outputs circuits in $\mathcal{C}$ for each irreducible factor, where $\mathcal{C}$ is some robust enough class of polynomials. 

However, \cref{thm:deterministic-factorization-constant-depth-circuits} appears to require \emph{blackbox} PITs for the class $\mathcal{C}'$, whereas \cite{KSS15} established such connections even in the whitebox setting. It is an intriguing open question if efficient whitebox algorithms for PIT of $\mathcal{C}'$ would imply efficient deterministic factoring algorithms for $\mathcal{C}$.

\section{Other applications}
\label{sec:applications}

\subsection{Hardness-randomness trade-offs for constant-depth circuits}

An immediate consequence of the closure theorems is that we get better hardness-randomness trade-offs for constant-depth circuits directly from the Kabanets-Impagliazzo hitting-set generator \cite{KI04}. 

\begin{theorem}[Hardness-randomness for constant-depth circuits]\label{thm:hardness-randomness-constant-depth}
    Let $\F$ be any field of characteristic $0$ or large enough. Fix any $\Delta > 0$. 
    Suppose there is an explicit family $\set{f_m(x_1,\ldots, x_m)}_{m \geq 0}$ of polynomials with $\deg(f_m) \leq m$ that requires depth $\Delta$ circuits of size $B(m)$ to compute them. Then, there is a family $\set{\mathcal{H}_n}$ of explicit hitting sets for the class polynomial size circuits of depth at most $\Delta - O(1)$ such that
    \[
    \abs{\mathcal{H}_n} = n^{O((B^{-1}(n))^2 / \log n)}.
    \]
    In particular, 
    \begin{itemize}\itemsep 0pt
    \item If $B(m) = 2^{\Omega(m)}$, then $\abs{\mathcal{H}_n} = n^{O(\log n)}$.
    \item If $B(m) = 2^{m^{\epsilon}}$ for some $\epsilon > 0$, then $\abs{\mathcal{H}_n} = n^{O(\log n)^c}$ for some $c > 0$.
    \item If $B(m) = m^{\omega(1)}$, then $\abs{\mathcal{H}_n} \leq n^{O(n^\epsilon)}$ for every $\epsilon > 0$.
\end{itemize}
\end{theorem}
\begin{proof}[Proof sketch]
The proof is exactly the same as the standard Kabanets-Impagliazzo generator for general circuits, except that instead of using Kaltofen's result for closure of general circuits under roots, we use \cref{thm:closure-powerseries-roots} instead. 
\end{proof}

\subsection{Border version of the factor conjecture}

Another consequence of the techniques in this paper is a conceptually simpler alternative proof of a result of B\"urgisser that shows that \emph{low degree} factors of polynomials with small circuits (but potentially exponentially high degree) are in the border of small circuits. We recall the formal theorem and discuss its proof below. 
\begin{theorem}[B\"urgisser~\cite{Burgisser04}] \label{thm:burgisser}
    Assume that $\operatorname{char}(\F) = 0$. 
    Suppose $P(\vecx) \in \F[\vecx]$ is computable by a size $s$ circuit (of possibly exponential degree) and $g$ is a factor of $P$. Then, the $\overline{\size}(g) = \poly(s, \deg(g))$. 
\end{theorem}
\begin{proof}
As in the previous cases, it would suffice to prove a $\overline{\size}$ upper bound for truncated power series roots of $P$. By working with a suitable shift, and the substitution $x_i \mapsto t x_i$, we can assume we have $P(t, y) \in \K[t,y]$ where $\K = \F(\vecx)$ with $\phi(t) \in \K\indsquare{t}$ satisfying $\phi(0) = 0$ and 
\[
P(t, y) = (y - \phi(t))^e \cdot (1 + Q(t,y))
\]
with $Q(0,0) = 0$. 

By \cref{thm:furstenberg}, we have 
\begin{align*}
    \phi(t) & = \diag\inparen{\frac{y^2 \cdot \partial_yP(ty,y)}{e \cdot P(ty,y)}} = \diag\inparen{\frac{y \cdot \partial_yP(ty,y)/y^{e-1}}{e \cdot P(ty,y)/y^{e}}}\\
    \implies \coeff{t^n}{\phi(t)} & = \coeff{t^ny^n}{\frac{y \cdot \partial_yP(ty,y)/y^{e-1}}{e \cdot P(ty,y)/y^{e}}}
\end{align*}
Since $P(t,y) = (y - \phi(t))^e \cdot (1 + Q(t,y))$, we have
\begin{align*}
\frac{P(ty, y)}{y^e} & = \inparen{1 - \frac{\phi(ty)}{y}}^e \cdot (1 + Q(ty,y)) = 1 - R(t,y)
\end{align*}
for some $R(t,y) \in \K\indsquare{t,y}$ with $R(0,0) = 0$.

Given a circuit $C$ for $P(t,y)$, we now have circuits $C_1(t,y), C_2(t,y)$ of $O(s)$ size such that $C_1$ has a single division by $y^{e-1}$ computing $\partial_yP(ty,y) / (e \cdot y^{e-1})$, and $C_2$ has a single division by $y^e$ and computes $R(t,y)$. 
Therefore,
\begin{align*}
    \coeff{t^ny^n}{\frac{y \cdot \partial_yP(ty,y)/y^{e-1}}{e \cdot P(ty,y)/y^{e}}} & = \coeff{t^ny^n}{\frac{y \cdot C_1}{1 - C_2}}\\
    & = \coeff{t^ny^n}{y \cdot C_1 \cdot \inparen{1 + C_2 + C_2^2 + \cdots}}\\
    & = \coeff{t^ny^n}{y \cdot C_1 \cdot \inparen{1 + C_2 + C_2^2 + \cdots + C_2^{2n}}}\\
    & =: \coeff{t^ny^n}{C_{3,n}(t,y)}
\end{align*}

Note that $C_{3,n}$ is a circuit of size $\poly(s,n)$ with divisions only by powers of $y$. By the border interpolation (\cref{lem:border-interpolation}), we can choose nonzero $\alpha^{(n)}_0, \ldots, \alpha^{(n)}_n, \beta^{(n)}_0, \ldots, \beta^{(n)}_n \in \F(\epsilon)$ such that 
\[
\sum_{i=0}^n \beta^{(n)}_i \cdot C_{3,n}(t, \alpha^{(n)}_i) = \coeff{y^n}{C_{3,n}(t,y)} + O(\epsilon)
\]
and the LHS is now a division-free circuit $C_{4,n}(t)$ of size $\poly(s,n)$. Once again,
\begin{align*}
    \sum_{i=0}^n \beta^{(n)}_i \cdot C_{4,n}(\alpha^{(n)}_i) & = \sum_{i,j=0}^n \beta^{(n)}_i \beta^{(n)}_j \cdot C_{3,n}(\alpha^{(n)}_i, \alpha^{(n)}_j) \\
    & = \coeff{t^n}{C_{4,n}(t)} + O(\epsilon) = \coeff{t^n y^n}{C_{3,n}(t,y)}+ O(\epsilon)
\end{align*}
Thus, we have a circuit $C_5(t) \in \K(\epsilon)[t]$ of size $\poly(n,s)$ defined by
\[
C_5(t) := \sum_{r=1}^n t^r \cdot \inparen{\sum_{i,j=0}^r \beta^{(r)}_i \beta^{(r)}_j \cdot C_{3,r}(\alpha^{(r)}_i, \alpha^{(r)}_j)}
\]
such that $C_5(t) = \homog_{\leq n}\inparen{\phi(t)} + O(\epsilon)$. Therefore, $\overline{\size}\inparen{\homog_{\leq n}\inparen{\phi(t)}} = \poly(s,n)$. 
\end{proof}

\section{Open questions}
\label{sec:open-problems}

We conclude with some open problems. 

\begin{enumerate}
    \item \textbf{Whitebox PIT to deterministic factorization:} Kopparty, Saraf and Shpilka~\cite{KSS15} showed that efficient algorithms for PIT for the class of general circuits leads to efficient deterministic factorization of general circuits, and this connection is for both the whitebox and the blackbox setting for PITs. Although \cref{cor:KSS-blackbox-generalisation} extends the blackbox connection to other natural subclasses of circuits (such as formulas, branching programs, constant-depth circuits), establishing a similar connection in the whitebox setting remains open. 

    \item \textbf{Computing $p$-th roots of circuits:} One of the simplest-to-state open problems in the area of factorization of algebraic circuits is the following --- over a characteristic $p$ field, if a polynomial $f^p$ is an $n$-variate, degree $d$ polynomial computed by a $\poly(n,d)$-sized circuit, is $f$ also computable by a $\poly(n,d)$-sized circuit? The answer to this question is unknown even for the setting of general algebraic circuits. 
    
\end{enumerate}
 \ifblind
\else
\paragraph*{Acknowledgements:} The discussions leading to this work started when a subset of the authors were at the workshop on Algebraic and Analytic Methods in Computational Complexity (Dagstuhl Seminar 24381) at Schloss Dagstuhl, and continued when they met again during the HDX \& Codes workshop at ICTS-TIFR in Bengaluru.  We are thankful to the organisers of these workshops and to the staff at these centers for the wonderful collaborative atmosphere that facilitated these discussions. 

Varun Ramanathan is grateful to Srikanth Srinivasan and Amik Raj Behera at the University of Copenhagen, and Nutan Limaye and Prateek Dwivedi at ITU Copenhagen, for the helpful discussions on polynomial factorization.
\fi


{\let\thefootnote\relax
\footnotetext{\textcolor{\gitinfonotecolour}{\gitinfonote \easteregg}
}}
\bibliographystyle{customurlbst/alphaurlpp}
\bibliography{crossref,references}

\appendix

\section{Extending closure results to fields of small characteristic}\label{sec:finite-fields-closure}

The closure results (\cref{thm:closure-powerseries-roots}, \cref{thm:closure-general-factors}) over zero or large characteristic fields extend to small characteristic fields, with some caveats. Suppose $P(\vecx) \in \F[\vecx]$ (where $\operatorname{char}(\F) = p$) has a constant-depth circuit, and let $g(\vecx)$ be any irreducible factor of $P(\vecx)$ with multiplicity $p^\ell e$ satisfying $\gcd(p, e) = 1$. Then, we show that $g(\vecx)^{p^\ell}$ has a constant-depth circuit over $\overline{\F}$, the algebraic closure of $\F$. These results follow due to a version of Furstenberg's theorem over small characteristic fields, which we state and prove below. Note that the case of roots of multiplicity 1 already follows from the original version of Furstenberg's theorem. In the following theorem we show how to extend it to higher-order multiplicity roots. 

\subsection{Furstenberg's theorem over small characteristic fields}

We shall work with the notion of Hasse derivative, which is the standard alternative to partial derivatives in the small characteristic setting. We state the definition and the product rule for Hasse derivatives. For more details, we recommend the reader to refer to \cite[Appendix C]{forbes-thesis-2014}.

\begin{definition}[Hasse derivatives]
    \label{defn:hasse-derivative}
    The \emph{Hasse Derivative of order $i$} of $F(t, y) \in \F[t, y]$ with respect to $y$, denoted as $\hasse{i}{y}(F)$, is defined as the coefficient of $z^i$ in the polynomial $F(t, y+z)$.
\end{definition}

\begin{lemma}[Product rule for Hasse derivatives] \label{lem:product-rule-hasse}
    Let $G(t, y), H(t,y) \in \F[t, y]$ be bivariate polynomials and let $k \geq 0$. Then,
    \[
    \hasse{k}{y}(GH) = \sum_{i+j=k} \hasse{i}{y}(G) \cdot \hasse{j}{y}(H)
    \]
\end{lemma}

The following version of Furstenberg's theorem over small characteristic is very similar to \cref{thm:furstenberg}, with some key differences. The theorem expresses an appropriate power of a power series root of a polynomial as a diagonal of a rational expression involving the polynomial and its derivatives.

\begin{theorem}[Furstenberg's theorem over small characteristic fields] \label{thm:furstenberg-small-characteristic}
    Let $\F$ be a field of characteristic $p$. Let $P(t,y) \in \F \indsquare{t,y}$ be a power series and $\varphi(t){\in \F\indsquare{t}}$ be a power series satisfying 
    \[
    P(t, y) = (y - \phi(t))^{p^\ell e} \cdot Q(t,y)
    \]
    for some $\ell \geq 0, e \geq 1$ such that $\gcd(p, e) = 1$. If $\phi(0) = 0$ and $Q(0,0) \neq 0$, then
    \begin{equation}
        \label{eqn: furstenberg-expression-small-char}
        \varphi^{p^\ell} = \diag\inparen{\frac{y^{2p^\ell} \cdot \hasse{p^\ell}{y}(P)(ty,y)}{e \cdot P(ty,y)}}
    \end{equation}
\end{theorem}
\begin{proof}
    Firstly, observe that 
    \begin{align*}
    \hasse{j}{y}((y - \phi(t))^{p^\ell e}) & = \coeff{z^j}{(y + z - \phi(t))^{p^\ell e}}\\
    & = \coeff{z^j}{(y^{p^\ell} + z^{p^\ell} - \phi(t)^{p^\ell})^e}
    \end{align*}
    Hence, $\hasse{j}{y}((y - \phi(t))^{p^\ell e}) = 0$ for all $0 < j < p^\ell$, and 
    \[
        \hasse{p^\ell}{y}((y - \phi(t))^{p^\ell e}) = e \cdot (y - \phi(t))^{p^\ell (e-1)}
    \]
    By applying product rule for Hasse derivatives (\cref{lem:product-rule-hasse}), $\hasse{p^\ell}{y}(P)(t, y)$ simplifies to
    \begin{align*}
        \hasse{p^\ell}{y}(P)(t, y) &= \sum_{i + j = p^\ell} \hasse{i}{y}((y - \phi(t))^{p^\ell e}) \cdot \hasse{j}{y}(Q)(t, y) \\
                                &= \hasse{p^\ell}{y}((y - \phi(t))^{p^\ell e}) \cdot \hasse{0}{y}(Q)(t, y) + \hasse{0}{y}((y - \phi(t))^{p^\ell e}) \cdot \hasse{p^\ell}{y}(Q)(t, y) \\
                                &= e \cdot (y - \phi(t))^{p^\ell (e-1)} \cdot Q(t, y) + (y - \phi(t))^{p^\ell e} \cdot \hasse{p^\ell}{y}(Q)(t, y)
    \end{align*}    
    Following along the lines of proof of \cref{thm:furstenberg},
    \begin{align*}
        P(t,y) &= (y-\varphi(t))^{p^\ell e} Q(t,y) \\
        \implies \frac{\hasse{p^\ell}{y}(P)(t,y)}{P(t,y)} & = \frac{e}{(y - \phi(t))^{p^\ell}} + \frac{\hasse{p^\ell}{y}(Q)(t,y)}{Q(t,y)}\\
        \implies \frac{y^{2p^\ell} \cdot \hasse{p^\ell}{y}(P)(ty,y)}{e \cdot P(ty,y)} & = \frac{y^{2p^\ell}}{(y - \phi(ty))^{p^\ell}} + \frac{y^{2p^\ell} \cdot \hasse{p^\ell}{y}(Q)(ty,y)}{e \cdot Q(ty,y)}\\
        \implies \diag\inparen{\frac{y^{2p^\ell} \cdot \hasse{p^\ell}{y}(P)(ty,y)}{e \cdot P(ty,y)}} & = \diag\inparen{\frac{y^2}{y - \phi(ty)}}^{p^\ell} + \diag\inparen{\frac{y^{2p^\ell} \cdot \hasse{p^\ell}{y}(Q)(ty,y)}{e \cdot Q(ty,y)}}. 
    \end{align*}
    As in the proof of \cref{thm:furstenberg}, the second term in the RHS is zero and 
    \[
    \diag\inparen{\frac{y^2}{y - \phi(ty)}} = \phi(t) \implies \diag\inparen{\frac{y^2}{y - \phi(ty)}}^{p^\ell} = \inparen{\phi(t)}^{p^\ell} \qedhere
    \]
\end{proof}
We can further simplify the expression in \autoref{thm:furstenberg-small-characteristic} to get a version of \autoref{cor:flajolet-soria-formula-for-roots} over small characteristic fields. 

\begin{corollary}[\cref{cor:flajolet-soria-formula-for-roots} for small characteristic]
    \label{cor:charp-flajolet-soria-formula-for-roots}
    Let $P(t,y), Q(t,y) \in \F\indsquare{t,y}$ and $\phi(t) \in \F\indsquare{t}$ satisfy
    \[
    P(t, y) = (y - \phi(t))^{p^\ell e} \cdot Q(t,y)
    \]
    with $\gcd(p,e) = 1$, $\phi(0) = 0$ and $Q(0,0) = \alpha \neq 0$. Then, 
    \[
    \phi(t)^{p^\ell} = \sum_{m\geq 0} \coeff{y^{p^{\ell}(e(m+1)-2)}}{\frac{\hasse{p^\ell}{y}(P)(t,y)}{e \cdot \alpha^{m+1}}{\inparen{\alpha y^{p^\ell\cdot e}-P(t,y)}^m}}.
    \]
    Moreover,
    \[
    \homog_{\leq d}[\phi(t)^{p^\ell}] = \homog_{\leq d}\insquare{\sum_{m\geq 0}^{2e(d+p^\ell) } \coeff{y^{p^{\ell}(e(m+1)-2)}}{\frac{\hasse{p^\ell}{y}(P)(t,y)}{e \cdot \alpha^{m+1}}{\inparen{\alpha y^{p^\ell\cdot e}-P(t,y)}^m}}}.
    \]
\end{corollary}
\begin{proof}
    By dividing $P$ by $\alpha$, let us assume without loss of generality that $Q(0,0) = 1$. 
    \begin{align*}
        \varphi^{p^\ell} &= \diag\inparen{\frac{y^{2p^\ell} \cdot \hasse{p^\ell}{y}(P)(ty,y)}{e \cdot P(ty,y)}}\\
        \coeff{t^k}{\varphi^{p^\ell}} &= \coeff{t^k y^k}{\frac{y^{2p^\ell} \cdot \hasse{p^\ell}{y}(P)(ty,y)}{e \cdot P(ty,y)}} = \coeff{t^k y^k}{\frac{y^{p^\ell(2-e)} \cdot \hasse{p^\ell}{y}(P)(ty,y)}{e \cdot P(ty,y) / y^{p^\ell\cdot e}}} \\
        &= \coeff{t^k y^k}{\frac{y^{p^\ell(2-e)} \cdot \hasse{p^\ell}{y}(P)(ty,y)}{e}{\sum_{m\geq 0}\inparen{1-\frac{P(ty,y)}{y^{p^\ell\cdot e}}}^m}}\\
        &= \coeff{t^ky^k}{\sum_{m\geq 1} \frac{y^{p^\ell(2-e(m+1))} \cdot \hasse{p^\ell}{y}(P)(ty,y)}{e}{\inparen{y^{p^\ell\cdot e}-P(ty,y)}^m}} \\
        &= \sum_{m\geq 0} \coeff{t^k y^{k - p^{\ell}(2-e(m+1))}}{\frac{\hasse{p^\ell}{y}(P)(ty,y)}{e}{\inparen{y^{p^\ell\cdot e}-P(ty,y)}^m}} \\
        &= \sum_{m\geq 0} \coeff{t^k y^{p^{\ell}(e(m+1)-2)}}{\frac{\hasse{p^\ell}{y}(P)(t,y)}{e}{\inparen{y^{p^\ell\cdot e}-P(t,y)}^m}} \\
        &= \coeff{t^k}{\sum_{m\geq 0} \coeff{y^{p^{\ell}(e(m+1)-2)}}{\frac{\hasse{p^\ell}{y}(P)(t,y)}{e}{\inparen{y^{p^\ell\cdot e}-P(t,y)}^m}}}\\ 
    \therefore \varphi^{p^\ell} & = \sum_{m\geq 0} \coeff{y^{p^{\ell}(e(m+1)-2)}}{\frac{\hasse{p^\ell}{y}(P)(t,y)}{e}{\inparen{y^{p^\ell\cdot e}-P(t,y)}^m}}
    \end{align*}
    The first part of the corollary follows from the above statement for the case of $Q(0,0) = 1$. \\

    \noindent
    To obtain a finite expression for $\homog_{\leq d}{\phi^{p^\ell}}$, observe that in 
    \[
    y^{p^\ell e} - P(t,y) = y^{p^\ell \cdot e} - \inparen{(y^{p^\ell} - \phi(t^{p^\ell}))^e \cdot (Q(t,y))}
    \]
    every monomial in $t$ has degree at least $p^\ell$. Furthermore, setting $t = 0$ reduces the above expression to $y^{p^\ell e} - y^{p^\ell e} \cdot Q(0,y)$
    which is divisible by $y^{p^{\ell} e + 1}$ (since $Q(0,0) = 1$). Therefore, 
    \[
    y^{p^\ell \cdot e} - P(t,y) = t^{p^\ell} \cdot A + y^{p^\ell e + 1} \cdot B
    \]
    for some $A,B \in \F\indsquare{t,y}$. Therefore, every monomial in $\inparen{y^{p^\ell\cdot e}-P(t,y)}^m$ with $t$-degree at most $d$ has $y$-degree at least $(m - \frac{d}{p^\ell}) \cdot (p^\ell e + 1)$, which is greater than $p^\ell (e(m+1)-2)$ when $m > d(e + \frac{1}{p^\ell}) + p^\ell (e-2)$. 
    Thus, for $m \geq 2e(d + p^\ell)$, there is no term with $t$-degree at most $d$ and $y$-degree $p^{\ell}(e(m+1)-2)$. Therefore,
    \[
    \homog_{\leq d}[\phi(t)^{p^\ell}] = \homog_{\leq d}\insquare{\sum_{m\geq 0}^{2e(d + p^\ell)} \coeff{y^{p^{\ell}(e(m+1)-2)}}{\frac{\hasse{p^\ell}{y}(P)(t,y)}{e \cdot \alpha^{m+1}}{\inparen{\alpha y^{p^\ell\cdot e}-P(t,y)}^m}}}.\qedhere
    \]
\end{proof}

\subsection{Complexity of power series roots and factors over {$\overline{\F_q}$}}

Using \cref{thm:furstenberg-small-characteristic} and \autoref{cor:charp-flajolet-soria-formula-for-roots}, we get the following analogue of \cref{thm:closure-powerseries-roots} over arbitrary fields of small characteristic.

\begin{theorem}[Power series roots with multiplicity over small characteristic]
    \label{thm:closure-powerseries-roots-small-char}
    Let $\F$ be a field of positive characteristic $p$. 
    Suppose $P(\vecx,y) \in \F[\vecx,y]$ is a polynomial computed by a circuit $C$, and $\varphi(\vecx) {\in \F\indsquare{\vecx}}$ is a power series satisfying $P(\vecx,y) = (y-\varphi(\vecx))^{p^\ell e} \cdot Q(\vecx,y)$ where $\varphi(\veczero) = 0$, $\gcd(p,e) = 1$ and $Q(\veczero,0)\neq 0$. Then, for any $d \in \N$, there is a circuit $C'$ over $\overline{\F}$ computing $\homog_{\leq d}\insquare{\varphi^{p^\ell}}$ such that
    \[\size(C') \leq \poly(d,\size(C))\]
    \[\depth(C') \leq \depth(C) + O(1)\]
\end{theorem}

Almost immediately, a similar statement follows for all factors of constant-depth circuits. Given a polynomial $P(\vecx)$, we apply a valid pre-processing map (\cref{defn:valid-pre}) to get $P(t,y)$ that is monic in $y$ and has the property that different power-series roots have different constant terms (a random shift suffices). We then apply \cref{thm:closure-powerseries-roots-small-char} to get small constant-depth circuits over $\overline{\F}$ for appropriate powers of each of the power-series roots, and then combine them (followed by a truncation) to get the following theorem.

\begin{theorem}[Complexity of factors over $\overline{\F_q}$]
    \label{thm:closure-factors-small-char-algebraic-closure}
    Let $\F_q$ be a field of positive characteristic $p$. Let $P(\vecx) \in \F_q[\vecx]$ be a polynomial on $n$ variables of degree $d$ computed by a circuit $C$ of size $s$ and depth $\Delta$. Further, let $g(\vecx)$ be a factor of $P(\vecx)$ with multiplicity $p^\ell\cdot e$ where $\gcd(p,e) = 1$. Then, $g(\vecx)^{p^\ell}$ is computable by a circuit of size $\poly(s,d,n)$ and depth $\Delta + O(1)$ over $\overline{\F_q}$. 
\end{theorem}

\end{document}